\documentclass[sigconf, nonacm]{acmart}

\newcommand\vldbdoi{XX.XX/XXX.XX}
\newcommand\vldbpages{XXX-XXX}
\newcommand\vldbvolume{14}
\newcommand\vldbissue{1}
\newcommand\vldbyear{2020}
\newcommand\vldbauthors{\authors}
\newcommand\vldbtitle{\shorttitle} 
\newcommand\vldbavailabilityurl{URL_TO_YOUR_ARTIFACTS}
\newcommand\vldbpagestyle{plain}

\usepackage{enumerate}
\usepackage{graphicx}
\usepackage{mathtools}
\usepackage{xspace}
\usepackage{subfigure}
\usepackage[lined,linesnumbered,ruled,commentsnumbered]{algorithm2e} 
\usepackage{amsthm}
\usepackage{hyperref}
\usepackage{caption}
\usepackage{float}
\usepackage{dsfont}
\usepackage[toc]{appendix}

\usepackage{tikz}
\usepackage{pgfplots}
\usepgfplotslibrary{patchplots} 
\pgfplotsset{compat=newest}
\usetikzlibrary{arrows,shapes,backgrounds,plotmarks,positioning}
\pgfplotsset{compat=newest}                         
\pgfplotsset{plot coordinates/math parser=false}

\newtheorem{theorem}{Theorem}
\newtheorem{lemma}{Lemma}

\newtheorem{proposition}{Proposition}

\newtheorem{example}{Example}
\newtheorem{remark}{Remark}

\newcommand*\dif{\mathop{}\!\mathrm{d}}

\newcommand{\D}{\mathcal{D}}
\newcommand{\SA}{\mathcal{S}}
\newcommand{\X}{\mathcal{X}}

\newcommand{\N}{\mathcal{V}}

\newcommand{\Real}{\mathds{R}}

\newcommand{\SetPair}{\mathds{S}}
\newcommand{\E}{\mathds{E}}

\newcommand{\Set}[1]{\{#1\}}

\newcommand{\supp}{\text{supp}}

\begin{document}
\title{Multi-user Pufferfish Privacy}

\author{Ni Ding}
\affiliation{%
  \institution{University of Auckland}
  \city{Auckland}
  \state{New Zealand}
}
\email{dingni529@gmail.com}

\author{Songpei Lu}
\affiliation{%
  \institution{Sophgo Technology Co., Ltd.}
  \city{Shenzhen}
  \country{China}
}
\email{songpei.lu@sophgo.com}


\author{Wenjing Yang}
\affiliation{%
  \institution{Beijing Institute of Technology}
  \city{Beijing}
  \country{China}
}
\email{wenjinyang@bit.edu.cn}

\author{Zijian Zhang}
\affiliation{%
  \institution{Beijing Institute of Technology}
	\city{Beijing}
	\country{China}
}
\email{zhangzijian@bit.edu.cn}


\begin{abstract}
This paper studies how to achieve individual indistinguishability by pufferfish privacy in aggregated query to a multi-user system. It is assumed that each user reports realization of a random variable. We study how to calibrate Laplace noise, added to the query answer, to attain pufferfish privacy when user changes his/her reported data value, leaves the system and is replaced by another use with different randomness. Sufficient conditions are derived for all scenarios for attaining statistical indistinguishability on four sets of secret pairs. They are derived using the existing Kantorovich method (Wasserstain metric of order $1$). These results can be applied to attain indistinguishability when a certain class of users is added or removed from a tabular data. It is revealed that attaining indifference in individual's data is conditioned on the statistics of this user only. For binary (Bernoulli distributed) random variables, the derived sufficient conditions can be further relaxed to reduce the noise and improve data utility. 
\end{abstract}

\maketitle

\pagestyle{\vldbpagestyle}
\begingroup\small\noindent\raggedright\textbf{PVLDB Reference Format:}\\
\vldbauthors. \vldbtitle. PVLDB, \vldbvolume(\vldbissue): \vldbpages, \vldbyear.\\
\href{https://doi.org/\vldbdoi}{doi:\vldbdoi}
\endgroup
\begingroup
\renewcommand\thefootnote{}\footnote{\noindent
This work is licensed under the Creative Commons BY-NC-ND 4.0 International License. Visit \url{https://creativecommons.org/licenses/by-nc-nd/4.0/} to view a copy of this license. For any use beyond those covered by this license, obtain permission by emailing \href{mailto:info@vldb.org}{info@vldb.org}. Copyright is held by the owner/author(s). Publication rights licensed to the VLDB Endowment. \\
\raggedright Proceedings of the VLDB Endowment, Vol. \vldbvolume, No. \vldbissue\ %
ISSN 2150-8097. \\
\href{https://doi.org/\vldbdoi}{doi:\vldbdoi} \\
}\addtocounter{footnote}{-1}\endgroup

\ifdefempty{\vldbavailabilityurl}{}{
\vspace{.3cm}
\begingroup\small\noindent\raggedright\textbf{PVLDB Artifact Availability:}\\
The source code, data, and/or other artifacts have been made available at \url{\vldbavailabilityurl}.
\endgroup
}

\section{Introduction}

In an era we require data streams are in nearly every interaction, safeguarding personal data is foundational to trust, compliance and ethical innovation. 
This rely on how personal privacy is defined or interpreted as to how to keep anonymous and prevent re-identification attacks. 
Differential privacy is proposed in a rigorous mathematical framework based on the statistical indistinguishability \cite{CalibNoiseDP,Wasserman2010,RDP2017}, meaning the difference in probability distribution of outputs from analyses on neighboring datasets (differing in only one entry/row) is upper bounded by a similarity threshold, parameterized by a privacy budge $\epsilon$.  
Assuming one record corresponding to a single individual \cite{Dong2022GDP,Dwork2014book}, $\epsilon$-differential privacy guarantees that the statistical resolution discriminating individuals' contribution and/or existence in a dataset is controlled below a threshold. This is a rigorous data protection against a strong adversary that is able to infer the knowledge of individual data by collecting probabilistic features of an observable output.

While the framework of differential privacy is designed for malicious statistical inference, it is assumed that the source data (e.g., the original database, or query answer) is deterministic,\footnote{Or, the database could have some (prior) randomness. But, it was not considered in differential privacy \cite{Dwork2014book}.}
i.e., each individual reports constantly the same value in a dataset. This might not be applicable to some cases where the participants' data are varying, e.g., local differential privacy \cite{LDP2013MiniMax}, sensitive surveys \cite{LensveltMulders2005}, private aggregation \cite{Gibbs2017} and federated analytics/learning \cite{Konevcny2016,Konevcny2016b}. 
To address this, \citet{Pufferfish2012KiferConf,Pufferfish2014Kifer} proposed a correlation-aware and adversary-aware privacy preserving framework called \emph{pufferfish privacy}, where the database is statistically conditioned on the secret/privacy, and the adversary obtains some prior belief. 
The purpose is to study a more practical scenario for data privacy in the presence of intrinsic randomness and to combat a more advanced statistical inference capability of the adversary. This framework not only extends differential privacy to more application-aware settings, but also establishes the connection to information leakage, a Bayesian style data privacy definition. 

For public data that statistically depends on some nesting secret, publishing it allows posterior information gain on secret, or privacy leakage, at the adversary side. This is captured by another field in computer science called quantitative information flow \cite{Smith2009,Clarkson2009,Alvim2012}, which has been frequently cited in recent works in privacy leakage \cite{Cuff2016,PvsInfer2012,Makhdoumi2014PF}. 
The idea is to measure privacy leakage as the difference between posterior and prior believes (on secret) \cite{Ding2021ISIT,MaxL2020,Liao2019_AlphaLeak}. 
As pointed out in \cite{Nuradha2023,PufferfishWasserstein2017Song,Ding2022AISTATS,Ding2024TIFS}, in the lens of optimization, it can be viewed as the mechanism design for attaining pufferfish privacy. 
This is why pufferfish privacy has been applied to attribute privacy~\cite{Zhang2020Attr,Zhang2022}, a typical data protection scenario when releasing tabular data by masking the sensitive columns. See also the experiments in \cite{Ding2022AISTATS,Ding2024TIFS}. 

While pufferfish privacy framework is generalized enough to be tailored to many data protection applications, privacy is still highly relevant to \emph{anonymization} or the conceal of personal ID, the fundamental motivation of both differential and pufferfish privacy \cite{Dwork2014book,Pufferfish2012KiferConf}. 
This is also based on the worst case concern: limiting a single person’s influence automatically restricts what can be learned about that person’s attributes. 
In fact, the proposal of pufferfish privacy in \cite{Pufferfish2012KiferConf} is based on the indistinguishability of modifying and removing an individual's record in a dataset, assuming the participants' data are probabilistic. 
While most studies focuses on mechanism design, e.g., \cite{RPP2024,PufferfishWasserstein2017Song,Ding2022AISTATS}, and theoretical analysis, e.g., \cite{Nuradha2023}, few discusses how to calibrate noise to prevent membership inference under the pufferfish privacy framework.
For example, \citet{PufferfishWasserstein2017Song,Ding2022AISTATS} discovered a Wasserstein metric approach to set up the Laplace mechanism. However, \citet{PufferfishWasserstein2017Song} focused on a time series secret that is Markovian, and \citet{Ding2022AISTATS} looked into tabular data release similar to attribute privacy. Question remains how to apply the stat-of-art noise calibration techniques to the original individual indistinguishability problem in \cite{Pufferfish2012KiferConf}.

This paper considers a multi-user system, where each user reporting a random variable to compute an aggregated query. We study the problem of how to achieve individual indistinguishability in terms of pufferfish privacy in the query answer. This is done by regarding statistical difference upon a small change in individual data, either its value or its statistical features. We study the scenarios including (1) a user changes his/her value reported to the query answer; (2) a user leaves the system, i.e., user's presence/absence; (3) a user is replaced by another one with different randomness. 
We also assume an attendance probability indicating the chances that a user will be present in the system. This aligns with the motivating example~\cite[Section~3.2]{Pufferfish2012KiferConf}, for the purpose of enforcing indistinguishability for almost all intrinsic statistics.
We adopt Laplace mechanism, i.e., adding Laplace noise to query answer, calibrated by Kantorovich optimal transport plan (corresponding to Wasserstein metric of order $1$) proposed in \cite{Ding2022AISTATS} to derive sufficient conditions for attaining pufferfish privacy for the three discriminative scenarios above.

Our main contributions are the following. 
\begin{enumerate}[1.]
	\item Denoting secret as the specific value reported by the user, pufferfish privacy on distinct secret pairs guarantees a certain degree of indifference in individual's data value revealed by the query answer. Similar to the $\ell_1$-sensitivity method \cite{CalibNoiseDP}, we prove that setting the scale parameter of Laplace noise by the distance between distinct secrets attains pufferfish privacy. Realizing the absence of an individual is equivalent of reporting zero, the same sufficient condition can be applied for ensuring enough indifference between reporting a specific value and being absent. 
	\item Denoting secret as the belief/fact that user's data follows certain probability distribution, we consider two scenarios: when this user is joining or leaving the system; when this user is replaced by another user. Both cases will cause a change in the resulting probability distribution of the query answer. We show that for both scenarios, sufficient conditions can be derived in terms of the user's statistical features, e.g., expectation, cumulative mass function, for attaining pufferfish privacy. In the case of binary random variable, relaxed sufficient conditions can be obtained to reduce the Laplace noise and improve data utility. 
	Experimental results on three real-world datasets indicates that pufferfish privacy in the above two scenarios can be applied to protect data privacy against removing, adding or modifying a certain class of users in a dataset. 
\end{enumerate}
The results above are produced by considering four secret pairs in the pufferfish privacy frameworks. It is revealed that for attaining the indifference about a specific user's data, the achievability of pufferfish privacy is conditioned on this user only. While the attendance probability is incorporated in the multi-user system, the derived sufficient conditions are independent of it.

\paragraph{Organization} The rest of paper is organized as follows. In Section~\ref{sec:preliminary}, we review the existing Wasserstein metric based noise calibration method and focus on the Kantorovich approach that will be used to derive the main results in this paper. 
Section~\ref{sec:sys} describes the multi-user system, aggregated query function, intrinsic statistics and secret statements considered in this system. 
Section~\ref{sec:PP} presents our main results: the sufficient conditions for attaining pufferfish privacy on different secret pairs, where numerous experimental results are displayed.

\paragraph{Notation}

We use capital letters to denote random variables (r.v.s) and lower case letters to denote the elementary event. Calligraphic letters refer to the alphabet of r.v.s. 
For example, $x$ is an instance of r.v. $X$, that takes value in alphabet $\X$.
Denote $P_{X}(x) = \Pr(X = x)$ the probability of outcome $X=x$ when r.v. $X$ takes the value $x$. We may use the notation $P_{X}(x)$ and $\Pr(X=x)$ interchangeably without confusion in this paper.
We use $P_X(\cdot)  = (P_{X}(x) \colon x \in \X)$ to denote a probability distribution and $X \sim P_{X}(\cdot)$ means that r.v. $X$ follows distribution $P_{X}(\cdot)$. 
The support of $P_{X}(\cdot)$  is denoted by $\supp(P_{X}) = \Set{x \in \X \colon  P_X(x) > 0}$.
The expected value of $f(X)$ for some deterministic function $f$ w.r.t. probability $P_{X}$ is denoted by $\E_{X \sim P_{X}}[f(X)] = \sum_{x \in \X} P_{X}(x) f(x)$.
The conditional probability $P_{Y|X}(y|x)  = \Pr(Y=y|X=x)$ denotes the chances of having $Y=y$ given the outcome $X=x$. $P_{Y|X}(\cdot|x)$ refers to the probability distribution of $Y$ conditioned on $X = x$.

\section{Preliminary}
\label{sec:preliminary}

We review the pufferfish privacy framework that was proposed by \citet{Pufferfish2012KiferConf,Pufferfish2014Kifer} and the existing noise calibration methods by Wasserstein metric. We will focused on the computationally efficient $W_1$ or Kantorovich approach for the Laplace mechanism proposed in~\cite{Ding2022AISTATS}.
 This is the mechanism adopted in this paper. 

\subsection{Pufferfish privacy} 
A data curator wants to publish data $X$ for legitimate uses, e.g., to convey the statistics of income tax of a population. But, the concern is that it may reveal some personal information, e.g., individual's race, education background and marital status, that can be exploited for adversarial  action, discrimination, etc. 
We call this part of personal information the sensitive data and denote by $S$. In general, $S$ is statistically correlated to $X$ so that a capable adversary would gain statistical knowledge on the sensitive data. This kind of privacy breach should be prevented and mitigated by the data curator before data release. 

Let $P_{X|S}(\cdot | s,\rho)$ denote the conditional probability of the public data $X$ given some secret $S = s$. This statistical dependence is due to the intrinsic correlation of the dataset. 
Here, $\rho$ is introduced as the adversaries prior knowledge as to the system settings, incl. the hyperparameters for parametric probability distribution, etc.. For example, the mean and covariance if $P_{X|S}(\cdot | s,\rho)$ for all $s \in \SA$ is Gaussian. There could be more than one adversaries in the system, and each of them has a different prior beliefs. 

The reason that we focus on the probability of $X$ given $S$ is to observe the difference between the statistics in the public data $X$ excited by two distinct secrets. To protect privacy, we want the difference between two probabilities $P_{X|S}(\cdot|s_i,\rho)$ and $P_{X|S}(\cdot|s_j,\rho)$ to be small so that an adversary observes $X$ would have difficulty in distinguishing whether it is resulted from $S = s_i$ or $S = s_j$. That is, the more statistical indistinguishability we have on the secret $S$, the more data privacy is preserved. 
This is fundamental motivation for differential privacy proposed by~\citet{Dwork2006}. However, when the original data $X$ is an r.v., pufferfish privacy, a more generalized framework, should be adopted.

Assume that a randomized version of $X$, denoted by $Y$, is generated to preserve data privacy. The adversary is assumed to have access to the privatized data $Y$ only. He can collect the aggregated statistics by repeatedly querying the released $Y$. 
Let $\SetPair$ be a subset of secret pairs $(s_i,s_j)$. A data curator can define $\SetPair$ to specify all secret pairs where statistical indistinguishability should be enforced to preserve data privacy. 
For a given \emph{privacy budget} $\epsilon>0$, the privatized data $Y$ is said to be $(\epsilon,\SetPair)$-\emph{pufferfish privacy} if \cite{Pufferfish2014Kifer, Pufferfish2012KiferConf}
\begin{align}\label{eq:pufferfish_privacy}
	e^{-\epsilon} \leq \frac{P_{Y|S}(y|s_i,\rho)}{P_{Y|S}(y|s_j,\rho)} \leq e^{\epsilon}, \quad \forall y,  (s_i,s_j) \in \mathbb{S}, \rho.
\end{align} 
Note that \eqref{eq:pufferfish_privacy} also guarantees $\epsilon$-degree of indistinguishability against all adversaries prior belief $\rho$.

\paragraph{\bf Laplace noise} 
One of the common method to randomize data is to add noise. Denote $N$ the noise that is independent of $X$. We have the randomized data generated by  $Y= X + N$. 
In this paper, we consider (zero-mean) Laplace noise $N_\theta$ with probability distribution 
\begin{equation}\label{eq:LapProb}
	P_{N_\theta}(z) = \frac{1}{2\theta}e^{-\frac{|z|}{\theta}}, \quad \forall z \in \Real. 
\end{equation}
For $X$ being an r.v., the resulting probability of $Y$ is given by a convolution
\begin{equation} \label{eq:conv}
	P_{Y|S}(y|s) = \int P_{N_\theta} (y-x) P_{X|S}(x|s).
\end{equation} 
The scale parameter $\theta$ indicates the flatness of Laplace distribution and determines its variance $\mathrm{VAR}[N_{\theta}] =  2\theta^2$. Therefore, $\theta$ specifies how large the Laplace noise is. 
Calibrating noise is equivalent to determining the value of $\theta$. A data curator needs to choose a proper $\theta$ to balance data privacy and utility. That is, the scale parameter $\theta$ should be just large enough to attain $(\epsilon,\SetPair)$-pufferfish privacy. An unnecessarily large $\theta$ will deteriorate the data utility.

\paragraph{\bf Kantorovich Mechanism} 

Knowing the intrinsic randomness in the pufferfish privacy framework, the literature suggests a noise calibration method in terms of statistical distances, which is analogous to the $\ell_1$-sensitivity method for differential privacy~\cite{CalibNoiseDP}. This can be seen from the initial attempts by~\citet{PufferfishWasserstein2017Song}, where $\theta$ is chosen as the maximum $\infty$-order Wasserstein distance between all secret pairs. 
To address a computation problem due to the non-convexity of the $\infty$-Wasserstein metric~\cite{Champion2008InfW,DePascale2019InfW}, \citet{Ding2022AISTATS} proposed a $1$-order Wasserstein (Kantorovich) approach. This will be the main mechanism we adopt in this paper. We review as follows.

For each pair of prior distributions $P_{X|S}(\cdot|s_i,\rho)$ and $P_{X|S}(\cdot|s_j,\rho)$, let $\pi$ be a \emph{coupling} joint distribution such that $P_{X|S}(x|s_i,\rho) = \int \pi(x,x') \dif x'$ for all $x$ and $P_{X|S}(x'|s_j,\rho) = \int \pi(x,x') \dif x$ for all $x'$. 
The $1$-Wasserstein distance is $W_1(s_i,s_j) := \inf_{\pi}\int |x-x'| \dif \pi(x,x^{\prime})$. The minimizer $\pi^*$ is called Kantorovich optimal transport plan~\cite{Villani2009OPT,Santambrogio2015OPT}..
The idea of calibrating noise for attaining pufferfish privacy is to search a $\theta$ such that $P_{Y|S}(y|s_i,\rho) \leq e^\epsilon P_{Y|S}(y|s_j,\rho)$. To do so, we have
\begin{align} \label{eq:W1:SuffCond1Main}
	P_{Y|S}&(y|s_i,\rho) - e^\epsilon P_{Y|S}(y|s_j,\rho) \nonumber \\
	& =  \int \big( P_{N_{\theta}}(y-x) - e^{\epsilon} P_{N_{\theta}}(y-x') \big) \dif \pi^*(x,x')  \nonumber\\
	& \leq \int P_{N_{\theta}}(y-x') \big( e^{\frac{|x-x'|}{\theta}}  - e^{\epsilon} \big)  \dif \pi^*(x,x'), \qquad \forall y
\end{align}
where the last inequality is due to the triangular inequality when we add Laplace noise. See \cite[Section~3.3]{Ding2022AISTATS}. 
Here, choosing $\pi^*$ against other couplings is for the purpose of searching a smaller value of $\theta$ that is pufferfish privacy attaining. 
Requesting $e^{\frac{|x-x'|}{\theta}}  - e^{\epsilon} \leq 0$ for all $(x,x') \in \X^2$ in \eqref{eq:W1:SuffCond1Main}, a sufficient condition for attaining pufferfish privacy is obtained below. 
\begin{proposition}[$W_1$ (Kantorovich) mechanism~{\cite[Lemma~1]{Ding2022AISTATS}}] \label{prop:Kantorovich}
	Adding Laplace noise $N_\theta$ with 	
	\begin{equation}\label{eq:w_1_mechanism}
		\theta = \frac{1}{\epsilon}\max_{\rho, (s_i, s_j) \in \mathbb{S}} \sup_{(x,x') \in \supp(\pi^*)} |x-x'|
	\end{equation}
	attains $(\epsilon, \SetPair)$-pufferfish privacy in $Y$. 
\end{proposition}
It is shown that the $W_1$ mechanism is equivalent to the $W_\infty$ mechanism proposed in \cite{PufferfishWasserstein2017Song}, but does not have computational problems. 
In fact, computing the Kantorovich optimal transport plan $\pi^*$ is much easier, which can be done directly using the exiting knowledge of $P_{X|S}(\cdot|s_i,\rho)$ and $P_{X|S}(\cdot|s_j,\rho)$: let $F_{X|S}(\cdot|s_i,\rho)$ and $F_{X|S}(\cdot|s_j,\rho)$ are the cumulative density functions for $P_{X|S}(\cdot|s_i,\rho)$ and $P_{X|S}(\cdot|s_j,\rho)$, respectively, 
\begin{equation} \label{eq:piStar}
	 \pi^*(x,x') = \frac{\dif^2}{\dif x  \dif x'} \min \big\{ F_{X|S}(x|s_i,\rho), F_{X|S}(x'|s_j,\rho) \big\}. 
\end{equation}
In fact, one does not need to computer $\pi^*$ to apply the Kantorovich mechanism \eqref{eq:w_1_mechanism}. It suffice to find the nonzero mass points $(x,x')$ to obtain $\sup_{(x,x') \in \supp(\pi^*)} |x-x'|$. We will propose such a method in Appendix~\ref{app1}, which will be applied to \autoref{sec:SuffCondSPQ} in this paper.

\section{Multi-user System}
\label{sec:sys}

Assume the database used for legitimate enquiery is contributed by a number of users (or user groups). We adopt a multi-user system setting in \cite[section~3.2]{Pufferfish2012KiferConf}, which is used to establish privacy with no assumptions, the situation when puffferfish privacy is equivalent to differential privacy. 

Denote $\N$ the finite set with each index $i \in \N$ denoting a unique user or participant. Let each user throw a dice $D_i$ independently and denote the outcome by a multiple random variable $D = (D_i \colon i \in \N)$ with alphabet $\D = \prod_{i \in \N} \D_i$. 
Here, $D_i$ may refer to a record of data constituted by different of fields/attributes of user $i$, and $D$ denotes a $|\N|$ size tabular database. Or, $D_i$ denotes scaler and $D$ is a $|\N|$-dimension random vector. 

Assume a data user keeps questioning the dataset $D$ and collect answers to gain knowledge on $D$. Let $f \colon \D \mapsto \Real$ be a query function and $X = f(D)$ be the query answer when it is applied to dataset $D$. 
For \emph{separable} $f$, we have 
\begin{align*}
	f(D) 
	&= \sum_{i \in V} f_i(D_i) 
\end{align*}
where we assume each user can choose its (possibly randomized) function $f_i$.\footnote{A randomized $f_i$ could denote the local data sanitation scheme chosen by individual $i$, as in local differential privacy \cite{LDP2013MiniMax}. However, instead of the design of $f_i$, we focus on the mechanism design after the data collection.}
If $D_i$ is an r.v., $f_i(D_i)$ is an r.v., too. This setting is equivalent to a summation query of $|V|$ r.v.s. Without loss of generality, we set $X = f(D)$ with $f$ being a deterministic summation function such that
$$ f(D) = \sum_{i \in \N} D_i $$
as  the choice of $f_i$ is not the focus of this paper. A data curator (server) only accepts what are reported by the users. 
Let $-i = \N \setminus \Set{i}$ denote the set containing all users except user $i$. We have $D = (D_i, D_{-i})$ where $D_{-i} = (D_j  \colon j \in V, j \neq i)$ such that 
$$ D_{-i} = f(D) - D_i. $$

For each user $i$, let $\zeta_i \in [0,1]$ be the presence probability such that $\Pr(\text{"user } i \text{ is present in } S \text{"})  = \zeta_i$ and $\Pr(\text{"user } i \text{ is absent in } S \text{"}) = 1-\zeta_i$. 
If user $i$ is present, the probability for him/her to report $D_i = a$ is denoted by $P_{i} (a) = \Pr (D_i = a)$ for all $a \in \D_i$. 
In this system, the database $D$ is an r.v. and the probability of $D$ is determined by the individual randomness (knowing $f$ is deterministic). 
For the prior belief $\rho$, it suffices to know the individual probabilities
$$ \rho = \Set{(\zeta_i, P_{i}(\cdot)) \colon i \in \N }$$
to determine the statistics of original $D$ and the query answer $X=f(D)$. 
For example, for $D = D_{-i}$ being the dataset instance where the $i$th user is missing but all other users $j$ report $a_j$, i.e., $D_{-i} = (a_j \colon j \in V, j \neq i )$.  
The probability of having this dataset is 
$$ \Pr(D) = (1-\zeta_i)  \cdot  \prod_{j\in\N \colon j \neq i} \zeta_j  P_{j} (a_j) $$
where there are only $|\N| - 1$ participants and we have 
$$
	P_X(x) 
	= \Pr( f(D_{-i}) = x) (1-\zeta_i) \prod_{j\in \N \colon j \neq i} \zeta_j.  
$$
Here, $f(D_{-i}) = \sum_{j \in \N \colon j \neq i} D_j$ is just a summation over $D_{-i}$.

\subsection{Secret}
Let the secret space $\SA$ contains all events where we want (or potentially want) to attain certain degree of statistical indistinguishability from others. 
Assume all other users remain in $D$, we consider different cases for user $i$ as follows. 
Let $s_{i,a}$ denote the event that user $i$ contribute a record/data $a$ to the dataset $D$ and $s_{i,\perp}$ denote the event that user $i$ is absent. We have the probabilities of two events being
\begin{align*}	
	& \Pr(S = s_{a_i}) = \zeta_i  P_{i}(a),  \\
	& \Pr(S = s_{\perp_i}) = 1-\zeta_i. 
\end{align*}
See \autoref{tab:ScrDesc}.
With the prior knowledge $\rho$, we can work out the probability of the query answer $X=f(D)$ conditioned on the secret $S$ as follows. 
\begin{align}
	P_{X|S}(x | s_{a_i}) 
	&= \Pr(f(D) = x | D_i = a_i) \cdot \prod_{j \in \N \colon j \neq i} \zeta_j \nonumber  \\
	&= \Pr(f(D_{-i}) = x - a_i) \cdot \prod_{j \in \N \colon j \neq i} \zeta_j , \label{eq:P_sa}\\
	P_{X|S}(x | s_{\perp_i}) 
	&=  \Pr(f(D_{-i}) = x )  \cdot \prod_{j \in \N \colon j \neq i} \zeta_j  \label{eq:P_sperp}
\end{align}
Here, the cases $s_{i,a}$ and  $s_{i,\perp}$ have all users statistically identical except the $i$th user. The latter case $s_{i,\perp}$ will have only $|\N|-1$ users in the system.

For a user that obtains an r.v. $D_i$, it is possible for $D_i$ to follow different distributions. We define $s_{i, P_{i}}$ the case that user $i$ is present and $D_i \sim P_{i}(\cdot)$. In this case,
\begin{align}
	P_{X|S}(x | s_{P_i}) 
	&= \int \Pr(f(D_{-i}) = x-t) P_i(t) \dif t \cdot  \prod_{j \in \N \colon j \neq i} \zeta_j \label{eq:P_sDi}
\end{align}

\begin{table}[t]
	\caption{Description of Secret $S$}
	\label{tab:ScrDesc}
	\begin{tabular}{cl}
		\toprule
		\hline
		\textbf{Secret Instance} & \textbf{Description} \\
		\midrule
		$s_{a_i}$ & User $i$ is present in $D$ and report value $a$.\\
		$s_{\perp_i}$ & User $i$ is absent in $D$.\\
		$s_{P_i}$ & User $i$ is present in $D$ and report r.v. $D_i$ that \\
		&  follows probability distribution $P_i(\cdot)$.\\
		\bottomrule
	\end{tabular}
\end{table}

\section{Pufferfish Privacy on Different $\SetPair$}
\label{sec:PP}

To protect privacy through pufferfish privacy framework, we need to first specify between which secret pairs certain degree of statistical indistinguishability should be be enforce. 
We first consider the following two sets of secret pairs.
\begin{align*}
	 \SetPair_{a,b} &= \Set{(s_{a_i}, s_{b_i}) \colon i \in \N}, \\
	  \SetPair_{a,\perp} &= \Set{(s_{a_i}, s_{\perp_i}) \colon i \in \N}.
\end{align*}
They are directly borrowed from differential privacy motivated by guaranteeing indistinguishability in the actual value of each records~\cite{CalibNoiseDP} and the existence of individual participants~\cite{Dwork2006}, respectively, in a dataset. For  secret pair set being $\SetPair_{a,b} \cup \SetPair_{a,\perp}$, we have the same system setting as "privacy with no assumption"~\cite[Section~3.2]{Pufferfish2012KiferConf}, which forms the basis of no-free-lunch privacy~\cite{Kifer2011NoFreeLunch}, where the statistical indistinguishability is enforced for almost all possible intrinsic statistics and therefore, as pointed out in ~\cite{Kifer2011NoFreeLunch}, the data utility could be severely damaged.

In addition to instance value of participants' data, pufferfish privacy framework is more powerful to ensure indistinguishability between individual statistics. 
Let $P_i(\cdot)$ and $Q_i(\cdot)$ be two distinct probability distributions over the alphabet $\D_i$. 
Adding and removing user $i$ will resulting in different statistics in query answers, the probability mass as well as the alphabet/support, while all secret pairs in $\SetPair_{a,b}$ and $\SetPair_{a,\perp}$ would only mean-shift the probability distribution; 
If $P_i(\cdot)$ and $Q_i(\cdot)$ are dissimilar, the resulting query answer probability will be obviously distinguishably recognizable by an adversarial observer, the situation we want to prevent for privacy protection. 
We consider the following two secret pair sets:
\begin{align*}
	\SetPair_{P,\perp} &= \Set{(s_{P_i}, s_{\perp_i}) \colon i \in \N }, \\
	\SetPair_{P,Q} &= \Set{(s_{P_i}, s_{Q_i}) \colon i \in \N}, 
\end{align*}
where $(s_{P_i}, s_{\perp_i})$ denotes the existence and nonexistence of user $i$ obtaining data following distribution $P_i(\cdot)$ and $\SetPair_{P,Q}$ denotes resolution whether user $i$ is emitting probability $P_i(\cdot)$ or $Q_i(\cdot)$. 

In this section, we first derive sufficient conditions for attaining pufferfish privacy in secret pair sets $\SetPair_{a,b}$ and $\SetPair_{a,\perp}$ and then $\SetPair_{P,\perp}$ and $\SetPair_{P,Q}$. 
These conditions are dependent of the $i$th user itself, not the rest users $-i$, in the system. 
We study the multi-user model in Section~\ref{sec:sys}, which incorporates individual user's presence probability $\zeta_i$. It will be shown that all sufficient conditions derived in this section do not depends on $\zeta_i$. 
We consider Laplace additive noise $Y = X + N_{\theta} =  f(D) + N_{\theta} $, where $P_{Y|S}(y|s)$ is determined by the convolution \eqref{eq:conv}.

\subsection{$(\epsilon, \SetPair_{a,b})$-Pufferfish Privacy}
\label{sec:Sab}

To guarantee indistinguishability between two distinct values that a user is reporting, it suffices to calibrate noise to the difference between these two values. 

\begin{theorem} \label{theo:SuffCondSab}
	For all $i \in \N$, given $a_i,b_i \in \D_i$ such that $a_i \neq b_i$, adding Laplace noise $N_\theta$ with 
	\begin{equation} \label{eq:SuffCondSab}
		\theta = \frac{\max_{i \in \N, \rho} | a_i - b_i |}{\epsilon}. 
	\end{equation}
	attains $(\epsilon, \SetPair_{a,b})$-pufferfish privacy in $Y$. 
\end{theorem}
\begin{proof}
		The proof is similar to \cite[Lemma~2(a)]{Ding2022AISTATS}, except we incorporate the presence probability $\zeta_i$. For $P_{X|S}(\cdot | s_{a_i})$ in~\eqref{eq:P_sa}, we have 
		\begin{align*}
			P_{Y|S}&(y|s_{a_i}) \\
			&= \int P_{N_\theta}(y-x) \Pr(f(D_{-i})=x-a_i) \prod_{j \in \N \colon j \neq i } \zeta_j \dif x \\
			&=  \int P_{N_\theta}( y-t-a_i ) \Pr(f(D_{-i})=t) \dif t \cdot \prod_{j \in \N \colon j \neq i } \zeta_j. 
		\end{align*}
		Then, 
		\begin{align}
			&P_{Y|S}(y|s_{a_i}) - e^{\epsilon} P_{Y|S}(y|s_{b_i})  \nonumber \\
			& = \int  \frac{1}{2\theta}  \Big(  e^{-\frac{|y-t-a_i|}{\theta}}  - e^{\epsilon-\frac{|y-t-b_i|}{\theta}} \Big) \Pr(f(D_{-i})=t) \dif t  \cdot \prod_{j \in \N \colon j \neq i } \zeta_j   \nonumber \\
			& = \int  \frac{1}{2\theta} e^{-\frac{|y-t-b_i|}{\theta}}  \cdot \Big(  e^{\frac{ |y-t-b_i| - |y-t-a_i|}{\theta}}  \nonumber \\
			&\qquad\qquad\qquad  - e^{\epsilon} \Big) \Pr(f(D_{-i})=t) \dif t \cdot \prod_{j \in \N \colon j \neq i } \zeta_j   \label{eq:TriInEq}\\
			& \leq   \int P_{N_{\theta}} (y-t-b_i)  \cdot \Big(  e^{\frac{ |a_i - b_i|}{\theta}}  - e^{\epsilon} \Big) \Pr(f(D_{-i})=t) \dif t \cdot \prod_{j \in \N \colon j \neq i } \zeta_j  \nonumber \\
			&= \Big(  e^{\frac{ |a_i -b_i|}{\theta}}  - e^{\epsilon} \Big) P_{Y|S}(y|s_{b_i}), \qquad \forall y,  \nonumber 
		\end{align}
		where \eqref{eq:TriInEq} is due to the triangular inequality. It is clear that $P_{Y|S}(y|s_{a_i}) \leq e^{\epsilon} P_{Y|S}(y|s_{b_i})$ for all $y$ if $\theta \geq \frac{| a_i - b_i |}{\epsilon}$. Maximize over all $i \in \N$ and prior belief $\rho$. Choose equality to get the minimal value of such $\theta$. We get \eqref{eq:SuffCondSab}.
\end{proof}

Theorem~\ref{theo:SuffCondSab} derives the same sufficient condition as \cite[Lemma~2(a)]{Ding2022AISTATS} where there is no presence probability, i.e., the deterministic states of existence of individual users. This means $(\epsilon,\SetPair_{a,b})$-pufferfish privacy does not depend on the value of $\zeta_i$.  

\begin{table}[h] 
	\caption{Statistics of Users 1, 2 and 3}
	\label{tab:ExpSetting}
	\begin{center}
		\begin{tabular}{ccccccc}
			\hline\hline
			& $X = 1$ & $X=2$ & $X=3$ & $X=4$ & $X=5$\\ \hline
			$\Pr(D_1 = \cdot)$ & 0.01  & 0.04 & 0.1 & 0.2 & 0.65 \\ \hline
			$\Pr(D_2 = \cdot)$ & 0.7  & 0.2 & 0.05 & 0.04 & 0.01 \\  \hline
			$\Pr(D_3 = \cdot)$ & 0.2  & 0.2 & 0.2 & 0.2 & 0.2 \\  \hline
		\end{tabular}
	\end{center}
\end{table}

Consider a $4$-user system with the first three users listed in \autoref{tab:ExpSetting}. 
Assuming the user $4$ data reporting changes from $D_i=5$ to $D_i=3$, we form a secret pair $(s_{a_4}, s_{b_4})$ with $a_4 = 5$ and $b_4 = 3$ for the $4$th position. 
By Theorem~\ref{theo:SuffCondSab}, $| a_4 - b_4 | = 2$ and so setting scale parameter tp $\theta = \frac{2}{\epsilon}$ attains pufferfish privacy. This can be seen from the optimal transport plan $\pi^*$ in Figure~\ref{fig:Sab}, where $\sup_{(x,x') \in \supp(\pi^*)}=2$ and by Proposition~\ref{prop:Kantorovich} we get $\theta = \frac{2}{\epsilon}$, too.

\begin{figure}[t] 
	\vspace*{-0.8cm}
	 \includegraphics[scale=0.6]{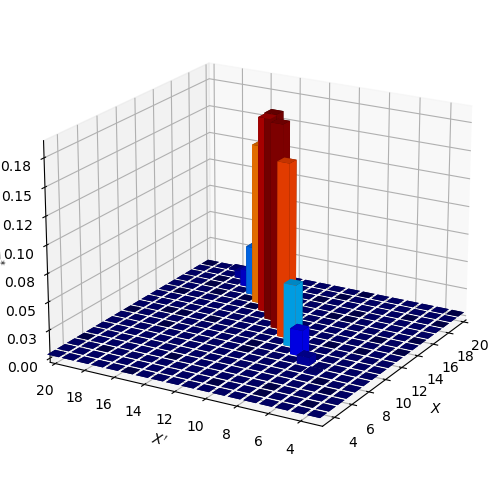}
	 \caption{Assume the first three users in Table~\ref{tab:ExpSetting} and consider $\SetPair_{a,b} = \Set{(s_{a_4}, s_{b_4})}$ for attaining distinguishability between user $4$ reporting $a_4 = 5$ and $b_4 = 3$. In this case, the Kantorovich optimal transport plan $\pi^*$ is shown in figure, where $\sup_{(x,x') \in \supp(\pi^*)}=2$. By Proposition~\ref{prop:Kantorovich}, adding Laplace noise with scale parameter $\theta = \frac{2}{\epsilon}$ attains $(\epsilon,\SetPair_{a,b})$-pufferfish privacy.}
	 \label{fig:Sab}
\end{figure}

\subsection{$(\epsilon, \SetPair_{a,\perp})$-Pufferfish Privacy}
\label{sec:Saperp}

Instead of the values each user reports, we consider the adversary's recognition of whether or not a user exists in the system by considering secret pair set $\SetPair_{a,\perp}$, where it is assumed user $i$ reports $a_i$ if she/he contributes to the database $D$. 
We show below that scale parameter of Laplace noise can be set by the value of $a_i$ for attaining pufferfish privacy.  

\begin{theorem} \label{theo:SuffCondSaperp}
	For all $a\in \D_i$, adding Laplace noise $N_\theta$ with 
	\begin{equation} \label{eq:SuffCondSaperp}
		\theta = \frac{1}{\epsilon} \max_{i \in \N, \rho} | a_i |. 
	\end{equation}
	attains $(\epsilon, \SetPair_{a,\perp})$-pufferfish privacy in $Y$. 
\end{theorem}
\begin{proof}
	By the expression $P_{X|S}(x|s_{\perp_i}) $ in~\eqref{eq:P_sperp}, we have 
	\begin{align*}
		P_{Y|S}(y|s_{\perp_i}) = \int P_{N_\theta}(y-x) \Pr(D_{-i})=x) \dif x \cdot \prod_{j \in \N \colon j \neq i } \zeta_j  
	\end{align*}
	and
	\begin{align*}
		P&_{Y|S}(y|s_{a_i}) - e^{\epsilon} P_{Y|S}(y|s_{\perp_i})  \nonumber \\
		&= \int  \frac{1}{2\theta}  \Big(  e^{-\frac{|y-t-a_i|}{\theta}} - e^{\epsilon-\frac{|y-t|}{\theta}} \Big) \Pr(f(D_{-i})=t) \dif t  \cdot \prod_{j \in \N \colon j \neq i } \zeta_j  \nonumber \\
		&\leq \int  \frac{1}{2\theta}  e^{-\frac{|y-t|}{\theta}} \Big( e^{-\frac{|a_i|}{\theta}} - e^\epsilon \Big)  \Pr(f(D_{-i})=t) \dif t  \cdot \prod_{j \in \N \colon j \neq i } \zeta_j   \\
		&=  \Big( e^{-\frac{|a_i|}{\theta}} - e^\epsilon \Big)  P_{Y|S} (y | s_{\perp_i}), \quad \forall y. 
	\end{align*}
	Therefore, it suffices to have \eqref{eq:SuffCondSaperp} such that $e^{-\frac{|a_i|}{\theta}} - e^\epsilon$ and so $P_{Y|S}(y|s_{a_i}) \leq e^{\epsilon} P_{Y|S}(y|s_{\perp_i}) $ for all $y$ to attain pufferfish privacy. 
\end{proof}

\begin{remark} \label{rem:SabSaperp}
	It is not difficult to see from the proof of \autoref{theo:SuffCondSaperp} that $(\epsilon, \SetPair_{a,\perp})$-pufferfish privacy is equivalent to $(\epsilon, \SetPair_{a,b})$-pufferfish privacy for $b_i = 0$. That is, for separable query function $f$, missing user $i$ is equivalent to that he/she report no numeric value to the database $D$. 
\end{remark}

Continue with the system settings in Figure~\ref{fig:Sab}, instead of user $4$ reports $b_4=3$, we assume he/she exits the system and compare the resulting statistics to the case when he/she is present and reporting $a_4=5$. 
By Theorem~\ref{theo:SuffCondSaperp}, setting $\theta = \frac{5}{\epsilon}$ attains $\epsilon$-pufferfish privacy on $(s_{a_4}, s_{\perp_4})$. 
This is verified by Figure~\ref{fig:Saperp}, where $\sup_{(x,x') \in \supp(\pi^*)}=5$. 

In Appendix~\ref{app:ex}, we show the computations of the Kantorovich optimal transport plan $\pi^*$ for Figures~\ref{fig:Sab} and \ref{fig:Saperp}. Also note that the results derive in Subsections~\ref{sec:Sab} and \ref{sec:Saperp} (on secret pairs $(s_{a_i},s_{b_i})$ and $(s_{a_i},s_{\perp_i})$) are independent of the $i$th user's statistics. This is because the secrets considered in $\SetPair_{a,b}$ and $\SetPair_{a,\perp}$ include deterministic events only. 

\begin{figure}[t]
	\vspace*{-0.8cm}
	\includegraphics[scale=0.6]{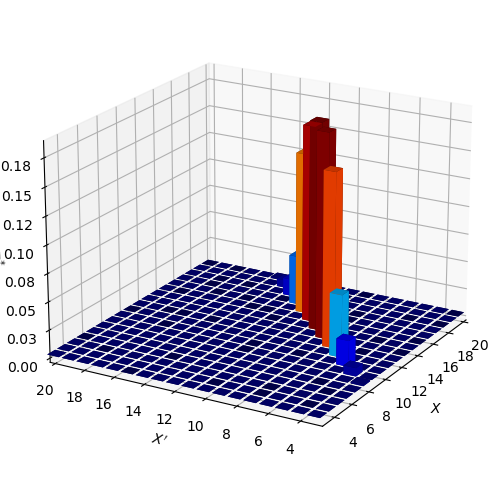}
	\caption{Assume the first three users in Table~\ref{tab:ExpSetting} and consider $\SetPair_{a,\perp} = \Set{(s_{a_4}, s_{\perp_4})}$ for attaining distinguishability between user $4$ existing and reporting $a_4 = 5$ and the absence of user $4$. The resulting Kantorovich optimal transport plan $\pi^*$ is shown in figure, where $\sup_{(x,x') \in \supp(\pi^*)}=5$. By Proposition~\ref{prop:Kantorovich}, adding Laplace noise with scale parameter $\theta = \frac{5}{\epsilon}$ attains $(\epsilon,\SetPair_{a,\perp})$-pufferfish privacy, which is equivalent to $(\epsilon,\SetPair_{a,0})$-pufferfish privacy (Remark~\ref{rem:SabSaperp}).}
	\label{fig:Saperp}
\end{figure}

\subsection{$(\epsilon, \SetPair_{P,\perp})$-Pufferfish Privacy}

We now pay attention to secrets that have impact on the resulting statistics in the query answer $f(D)$. We first study the probabilistic difference between the existence and non-existence of a user that obtains an r.v. $D_i$. Theorem below shows that pufferfish privacy attainability depends on the statistics of $D_i$.

\begin{theorem} \label{theo:SuffCondSPperp}
	Adding Laplace noise $N_\theta$ attains $(\epsilon, \SetPair_{P,\perp})$-pufferfish privacy in $Y$ if $\theta$ satisfies
	\begin{equation} \label{eq:SuffCondSPperpRelax}
			\E_{D_i \sim P_{i}(\cdot)} \big[e^{\frac{|D_i|}{\theta}} \big] = e^{\epsilon}
	\end{equation}
	for all $i \in \N$ and $\rho$, or
	\begin{equation} \label{eq:SuffCondSPperp}
		\theta = \max\limits_{t\in \D_i \colon i \in \N, \rho}\frac{|t|}{\epsilon}
	\end{equation}
\end{theorem}
\begin{proof}
	By  \eqref{eq:P_sDi}, we have 
	\begin{align*}
		&P_{Y|S}(y|s_{P_{i}}) \\
		& = \int  \Big( \int P_{N_\theta}(y-m-t) \Pr(f(D_{-i})=m) \dif m \Big) P_{i}(t) \dif t  \prod_{j \in \N \colon j \neq i} \zeta_j
	\end{align*}
	and then
	\begin{align*}
		&P_{Y|S}(y|s_{i,P_{i}}) - e^{\epsilon} P_{Y|S}(y|s_{\perp_i})  \\
		& =  \int  \Big( \int \big( P_{N_\theta}(y-m-t)  - e^{\epsilon} P_{N_\theta}(y-m) \big) \cdot \\
		& \qquad\qquad\qquad  \Pr(f(D_{-i})=m) \dif m \Big) P_{i}(t) \dif t  \prod_{j \in \N \colon j \neq i}  \\
		&= \int \Big(  \int \frac{1}{2\theta} \big( e^{-\frac{|y-m-t|}{\theta}} - e^{\epsilon-\frac{|y-m|}{\theta}}\big)  \cdot \\
		& \qquad\qquad\qquad  \Pr(f(D_{-i})=m) \dif m \Big) P_{i}(t) \dif t \prod_{j \in \N \colon j \neq i} \\
		&\leq \int  \big( e^{\frac{|t|}{\theta}} - e^\epsilon \big) \int P_{N_\theta}(y-m) \cdot \\
		& \qquad\qquad\qquad  \Pr(f(D_{-i})=m) \dif m \cdot P_{i}(t) \dif t \prod_{j \in \N \colon j \neq i}  \\
		&= P_{Y|S}(y|s_{i,\perp}) \int \big( e^{\frac{|t|}{\theta}} - e^\epsilon \big) P_{D_i}(t) \dif t, \quad \forall y
	\end{align*}
	So, $\int e^{\frac{|t|}{\theta}} P_{i}(t) \dif t =	\E_{D_i \sim P_{i}(\cdot)} \big[e^{\frac{|D_i|}{\theta}} \big] \leq e^{\epsilon }$ is a sufficient condition to hold $P_{Y|S}(y|s_{P_{i}}) \leq e^{\epsilon} P_{Y|S}(y|s_{\perp_i})$ for all $y$. Therefore, we have condition~\eqref{eq:SuffCondSPperpRelax} for attaining $(\epsilon, \SetPair_{P,\perp})$-pufferfish privacy. 
	Knowing that $\E_{D_i \sim P_{i}(\cdot)} \big[e^{\frac{|D_i|}{\theta}} \big] \leq \max_{t \in \D_i} e^{\frac{|t|}{\theta}} $, it suffices to have 
	$$
		\max_{t \in \D_i} e^{\frac{|t|}{\theta}} \leq e^{\epsilon} \quad \Longrightarrow \quad 
		\theta \geq \max_{t \in \D_i} \frac{|t|}{\epsilon }
	$$
	in order to hold~\eqref{eq:SuffCondSPperpRelax}. Take the maximization further over all $i \in \N$ and $\rho$, and reduce to equation to get the minimum such $\theta$, we get~\eqref{eq:SuffCondSPperp}. 
\end{proof}

\begin{remark} \label{rem:SuffCondSPperp}
	The proof of \autoref{theo:SuffCondSPperp} indicates that sufficient condition \eqref{eq:SuffCondSPperpRelax} is more relaxed than \eqref{eq:SuffCondSPperp}, which will produce a smaller value of $\theta$. 
\end{remark}

\begin{figure}[t] 
	\vspace*{-0.8cm}
	\includegraphics[scale=0.6]{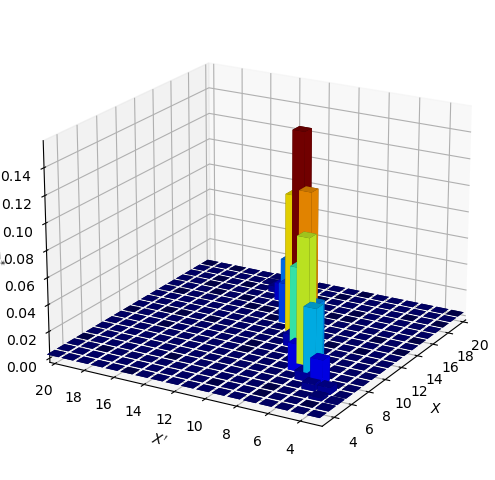}
	\caption{For the first three users in Table~\ref{tab:ExpSetting}, consider adding a $4$th user such that $D_4 \sim P_4(\cdot)$ in Table~\ref{tab:ExpSettingExtraTwo}. For secret pair $\SetPair_{P,\perp} = \Set{(s_{P_4}, s_{\perp_4})}$, the resulting Kantorovich optimal transport plan $\pi^*$ is shown in figure, where $\sup_{(x,x') \in \supp(\pi^*)}=5$. By Proposition~\ref{prop:Kantorovich}, adding Laplace noise with scale parameter $\theta = \frac{5}{\epsilon}$ attains $(\epsilon,\SetPair_{a,\perp})$-pufferfish privacy.}
	\label{fig:SPperp}
\end{figure}

\begin{table}[t]
	\caption{Two choices of Users 4}
	\label{tab:ExpSettingExtraTwo}
	\begin{center}
		\begin{tabular}{ccccccc}
			\hline\hline
			& $X = 1$ & $X=2$ & $X=3$ & $X=4$ & $X=5$\\ \hline
			$P_4(\cdot)$ & 0.4  & 0.1 & 0 & 0.1 & 0.4 \\ \hline
			$Q_4(\cdot)$ & 0  & 0.05 & 0.9 & 0.05 & 0 \\  \hline
		\end{tabular}
	\end{center}
\end{table}

For the 3-user system in Table~\ref{tab:ExpSetting}, we add the $4$th user with $D_4 \sim P_4(\cdot)$ in Table~\ref{tab:ExpSettingExtraTwo}. It is clear that $\max_{t\in \D_4} |t| = 5$ and therefore by \eqref{eq:SuffCondSPperp} in Theorem~\ref{theo:SuffCondSPperp}, adding Laplace noise with $\theta = \frac{|5}{\epsilon}$ attains pufferfish privacy. This result coincide with the one following the Kantorovich approach in Proposition~\ref{prop:Kantorovich}. See Figure~\ref{fig:SPperp}.

\begin{figure}[t]
	\centerline{\scalebox{0.6}{
\begin{tikzpicture}
	
	\definecolor{darkgray176}{RGB}{176,176,176}
	\definecolor{darkorange25512714}{RGB}{255,127,14}
	\definecolor{steelblue31119180}{RGB}{31,119,180}
	
	\begin{axis}[
		width=4in,
		height=2.5in,
		scale only axis,
		xmin=0.1,
		xmax=1,
		xlabel={\Large privacy budget $\epsilon$},
		ymin=0,
		ymax=50,
		ylabel={\Large $\theta$},
		grid=major,
		legend style={at={(0.98,0.93)},draw=darkgray!60!black,fill=white,legend cell align=left}
		]

		\addplot [
		color=blue,
		line width = 1.5pt]
		table {%
			0.1 50
			0.118367346938776 42.2413793103448
			0.136734693877551 36.5671641791045
			0.155102040816327 32.2368421052632
			0.173469387755102 28.8235294117647
			0.191836734693878 26.063829787234
			0.210204081632653 23.7864077669903
			0.228571428571429 21.875
			0.246938775510204 20.2479338842975
			0.26530612244898 18.8461538461538
			0.283673469387755 17.6258992805755
			0.302040816326531 16.5540540540541
			0.320408163265306 15.6050955414013
			0.338775510204082 14.7590361445783
			0.357142857142857 14
			0.375510204081633 13.3152173913043
			0.393877551020408 12.6943005181347
			0.412244897959184 12.1287128712871
			0.430612244897959 11.6113744075829
			0.448979591836735 11.1363636363636
			0.46734693877551 10.6986899563319
			0.485714285714286 10.2941176470588
			0.504081632653061 9.91902834008097
			0.522448979591837 9.5703125
			0.540816326530612 9.24528301886792
			0.559183673469388 8.94160583941606
			0.577551020408163 8.65724381625442
			0.595918367346939 8.39041095890411
			0.614285714285714 8.13953488372093
			0.63265306122449 7.90322580645161
			0.651020408163265 7.68025078369906
			0.669387755102041 7.46951219512195
			0.687755102040816 7.2700296735905
			0.706122448979592 7.08092485549133
			0.724489795918367 6.90140845070423
			0.742857142857143 6.73076923076923
			0.761224489795918 6.56836461126005
			0.779591836734694 6.41361256544503
			0.797959183673469 6.26598465473146
			0.816326530612245 6.125
			0.83469387755102 5.99022004889976
			0.853061224489796 5.86124401913876
			0.871428571428571 5.73770491803279
			0.889795918367347 5.61926605504587
			0.908163265306122 5.50561797752809
			0.926530612244898 5.39647577092511
			0.944897959183673 5.29157667386609
			0.963265306122449 5.19067796610169
			0.981632653061225 5.09355509355509
			1 5
		};
		\addlegendentry{\Large $\theta$ by \eqref{eq:SuffCondSPperp} in Theorem~\ref{theo:SuffCondSPperp}};
		
		\addplot [
		color=red,
		dashed,
		line width = 1.5pt]
		table {%
			0.1 30.5560385049098
			0.118367346938776 25.8989317567273
			0.136734693877551 22.4924742969577
			0.155102040816327 19.8923587801101
			0.173469387755102 17.8424551460408
			0.191836734693878 16.1847257453729
			0.210204081632653 14.8163694284809
			0.228571428571429 13.6676283060223
			0.246938775510204 12.6894989305344
			0.26530612244898 11.8465480642685
			0.283673469387755 11.1125194438068
			0.302040816326531 10.4675434820174
			0.320408163265306 9.89630669044798
			0.338775510204082 9.38681659090417
			0.357142857142857 8.92954774010077
			0.375510204081633 8.51683837758175
			0.393877551020408 8.14245588723585
			0.412244897959184 7.80127842262237
			0.430612244897959 7.48905801212099
			0.448979591836735 7.20224181104611
			0.46734693877551 6.9378355039592
			0.485714285714286 6.69329769977258
			0.504081632653061 6.46645741459801
			0.522448979591837 6.25544896058928
			0.540816326530612 6.05866010274797
			0.559183673469388 5.87469043302826
			0.577551020408163 5.70231768721652
			0.595918367346939 5.54047029090411
			0.614285714285714 5.38820483078816
			0.63265306122449 5.24468745034572
			0.651020408163265 5.10917839484713
			0.669387755102041 4.98101910080441
			0.687755102040816 4.85962135418761
			0.706122448979592 4.74445814072537
			0.724489795918367 4.63505588800364
			0.742857142857143 4.53098785847357
			0.761224489795918 4.43186849897855
			0.779591836734694 4.33734858904946
			0.797959183673469 4.24711105926551
			0.816326530612245 4.16086737414508
			0.83469387755102 4.07835439260904
			0.853061224489796 3.99933163403766
			0.871428571428571 3.92357889007877
			0.889795918367347 3.85089413224764
			0.908163265306122 3.78109167344188
			0.926530612244898 3.71400054813625
			0.944897959183673 3.64946308150152
			0.963265306122449 3.58733362223044
			0.981632653061225 3.52747741762885
			1 3.46976961268028
		};
		
		\addlegendentry{\Large $\theta$ by \eqref{eq:SuffCondSPperpRelax} in Theorem~\ref{theo:SuffCondSPperp} using Brent; method};
	\end{axis}
	
\end{tikzpicture}}}
	\caption{For the experiment in Figure~\ref{fig:SPperp}, we use Brent method to approximate a $\theta$ that satisfies $\E_{D_i \sim P_{D_i}} \big[e^{\frac{|D_i|}{\theta}} \big] \leq e^{\epsilon}$ by \eqref{eq:SuffCondSPperpRelax}. Compared to \eqref{eq:SuffCondSPperp}, the value of $\theta$ is reduced, which indicates a noise reduction, i.e., an improvement in data utility, for attaining pufferfish privacy. }
	\label{fig:Brent}
\end{figure}
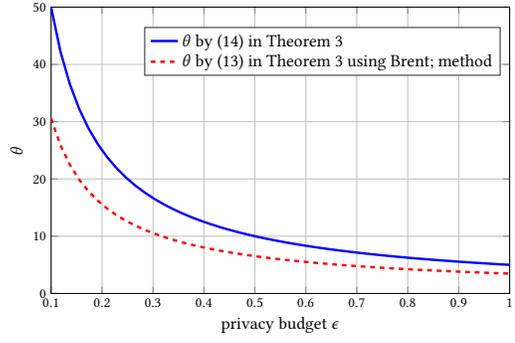

\subsubsection{Brent's method}
Note that the noise calibration method~\eqref{eq:SuffCondSPperpRelax} requires solving a polynomial equation. While in general case an exact solution is not feasible, one can apply numerical methods. Figure~\ref{fig:Brent} shows the $\theta$ value satisfying $\E_{D_i \sim P_{D_i}} \big[e^{\frac{|D_i|}{\theta}} \big] \leq e^{\epsilon}$ by~\eqref{eq:SuffCondSPperpRelax} approximated by the Brent's method \cite{Brent1971,Suli2003book}. It verifies Remark~\ref{rem:SuffCondSPperp} that \eqref{eq:SuffCondSPperpRelax} is a relaxed sufficient condition. It reduces the value of $\theta$ and improves data utility. 

An interesting observation is that an r.v. $D_i$ can be think of a number of records in a dataset that belong to a certain class, e.g., the teenagers, diabetic patients. Each record is considered a realization of $D_i$, and so the absence of $D_i$ refers to removing all records belonging to this class. 
That is, the pufferfish privacy is able to guarantee indistinguishability in changing single record, as well as a class of records generated by a subgroup of population.

To see this, consider three real-world datasets in the UCI machine learning repository~\cite{UCI2007}: \texttt{adult}, \texttt{bank marketing} and \texttt{student performance}. 
For dataset \texttt{adult}, we use the \texttt{education} field for all records such that \texttt{race}$=$`\texttt{white}' as $D_i$, i.e., $D_i  \sim P_i(\cdot) = \Pr(\texttt{education} = \cdot | \texttt{race} = \texttt{`white'})$. Consider when removing all records as realizations of $D_i$ (i.e., deleting all \texttt{white} people's data) from the dataset. The resulting query probability will be different and we apply \eqref{eq:SuffCondSPperp} and \eqref{eq:SuffCondSPperpRelax} by Brent's method in Theorem~\ref{theo:SuffCondSPperp}  to choose $\theta$ for attaining pufferfish privacy. We plot both $\theta$'s versus privacy budget $\epsilon$ in Figure~\ref{fig:SuffCondUCIPperp}(a), where the noise incurred by \eqref{eq:SuffCondSPperpRelax}  is much smaller than \eqref{eq:SuffCondSPperp}.  
We do the same experiment on \texttt{student performance} and \texttt{bank marketing} by letting $D_i  \sim P_i (\cdot) = \Pr(\texttt{romatic} = \cdot | \texttt{higher} = \texttt{`yes'})$ and $D_i  \sim P_i(\cdot) = \Pr(\texttt{marital} = \cdot | \texttt{loan} = \texttt{`yes'})$, respectively. We get the same results in Figures~\ref{fig:SuffCondUCIPperp}(b) and (c).

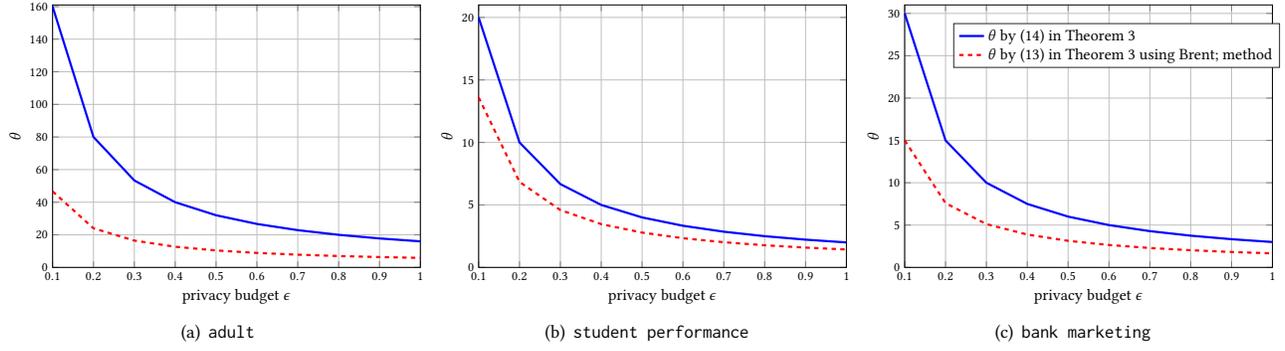
\begin{figure*}[t]
	\subfigure[\texttt{adult}]{\scalebox{0.55}{
\begin{tikzpicture}

\definecolor{darkgray176}{RGB}{176,176,176}
\definecolor{darkorange25512714}{RGB}{255,127,14}
\definecolor{steelblue31119180}{RGB}{31,119,180}

\begin{axis}[
width=3.5in,
height=2.5in,
scale only axis,
xmin=0.1,
xmax=1,
xlabel={\Large privacy budget $\epsilon$},
ymin=0,
ymax=161,
ylabel={\Large $\theta$},
grid=major,
legend style={at={(0.98,0.93)},draw=darkgray!60!black,fill=white,legend cell align=left}
]

\addplot [
color=blue,
line width = 1.5pt
]
table {%
0.1 160
0.2 80
0.3 53.3333333333333
0.4 40
0.5 32
0.6 26.6666666666667
0.7 22.8571428571429
0.8 20
0.9 17.7777777777778
1 16
};

\addplot [
color=red,
dashed,
line width = 1.5pt]
table {%
	0.1 46.6282013604038
	0.2 24.0038124742262
	0.3 16.458016985655
	0.4 12.681317723579
	0.5 10.4119029270012
	0.6 8.89589797438197
	0.7 7.81025935004597
	0.8 6.99349894850444
	0.9 6.35592622322652
	1 5.84374564968718
};
\end{axis}

\end{tikzpicture}}}
	\subfigure[\texttt{student performance}]{\scalebox{0.55}{
\begin{tikzpicture}

\definecolor{darkgray176}{RGB}{176,176,176}
\definecolor{darkorange25512714}{RGB}{255,127,14}
\definecolor{steelblue31119180}{RGB}{31,119,180}

\begin{axis}[
width=3.5in,
height=2.5in,
scale only axis,
xmin=0.1,
xmax=1,
xlabel={\Large privacy budget $\epsilon$},
ymin=0,
ymax=21,
ylabel={\Large $\theta$},
grid=major,
legend style={at={(0.98,0.93)},draw=darkgray!60!black,fill=white,legend cell align=left}
]
\addplot [
color=blue,
line width = 1.5pt
]
table {%
0.1 20
0.2 10
0.3 6.66666666666667
0.4 5
0.5 4
0.6 3.33333333333333
0.7 2.85714285714286
0.8 2.5
0.9 2.22222222222222
1 2
};

\addplot [
color=red,
dashed,
line width = 1.5pt]
table {%
	0.1 13.6016549298052
	0.2 6.84306843443145
	0.3 4.5901927683633
	0.4 3.46371874229823
	0.5 2.78778553376052
	0.6 2.33710710635313
	0.7 2.01513311265748
	0.8 1.77358910857529
	0.9 1.58565661565653
	1 1.43524513878258
};

\end{axis}

\end{tikzpicture}}}
	\subfigure[\texttt{bank marketing}]{\scalebox{0.55}{
\begin{tikzpicture}

\definecolor{darkgray176}{RGB}{176,176,176}
\definecolor{darkorange25512714}{RGB}{255,127,14}
\definecolor{steelblue31119180}{RGB}{31,119,180}

\begin{axis}[
width=3.5in,
height=2.5in,
scale only axis,
xmin=0.1,
xmax=1,
xlabel={\Large privacy budget $\epsilon$},
ymin=0,
ymax=31,
ylabel={\Large $\theta$},
grid=major,
legend style={at={(1.02,0.93)},draw=darkgray!60!black,fill=white,legend cell align=left}
]

\addplot [
color=blue,
line width = 1.5pt
]
table {%
0.1 30
0.2 15
0.3 10
0.4 7.5
0.5 6
0.6 5
0.7 4.28571428571429
0.8 3.75
0.9 3.33333333333333
1 3
};
\addlegendentry{\Large $\theta$ by \eqref{eq:SuffCondSPperp} in Theorem~\ref{theo:SuffCondSPperp}};

\addplot [
color=red,
dashed,
line width = 1.5pt]
table {%
	0.1 15.0093546626372
	0.2 7.59101451165105
	0.3 5.11870868042967
	0.4 3.88281778181175
	0.5 3.14142069372006
	0.6 2.64721324908633
	0.7 2.29421078723952
	0.8 2.02942369376946
	0.9 1.82341583498072
	1 1.65852786984989
};
\addlegendentry{\Large $\theta$ by \eqref{eq:SuffCondSPperpRelax} in Theorem~\ref{theo:SuffCondSPperp} using Brent; method};

\end{axis}

\end{tikzpicture}}}
	\caption{Two values of $\theta$ produced by \eqref{eq:SuffCondSPperpRelax} and \eqref{eq:SuffCondSPperp} in Theorem~\ref{theo:SuffCondSPperp} for achieving $\epsilon$-pufferfish privacy on $\SetPair_{P,\perp} = \Set{(P_i,\perp_i)}$ on three data sets in the UCI machine learning repository, where $P_i (\cdot) = \Pr(\texttt{education} = \cdot | \texttt{race} = \texttt{`white'})$, $P_i(\cdot)  = \Pr(\texttt{romatic} = \cdot | \texttt{higher} = \texttt{`yes'})$ and $P_i(\cdot)  = \Pr(\texttt{marital} = \cdot | \texttt{loan} = \texttt{`yes'})$ for dataset \texttt{adult}, \texttt{student performance} and \texttt{bank marketing}, respectively.  
	}
	\label{fig:SuffCondUCIPperp}
\end{figure*}

\subsubsection{Binary $D_i$}
%

We also show an example of Remark~\ref{rem:SuffCondSPperp} in a binary r.v. case, where the value of $\theta$ in~\eqref{eq:SuffCondSPperpRelax} can be worked out in a closed form. 
For Bernoulli distributed $D_i \sim \text{Bernoulli}(p_i)$ for all $i \in \N$ such that $\Pr(D_i = 1) = p_i$ and $\Pr(D_i = 0) = 1-p_i$, we have 
$\E_{D_i \sim P_{D_i}} \big[e^{\frac{|D_i|}{\theta}} \big] = p_i e^{\frac{1}{\theta}} + (1-p_i) = e^{\epsilon}$. This gives
\begin{equation}  \label{eq:SuffCondSPperpRelaxBinary}
	\theta \geq \frac{1}{\log \frac{e^{\epsilon} - (1-p_i)}{p_i}}
\end{equation}
However, we will get a sufficient condition
$
	\theta = \frac{1}{\epsilon}, 
$
a value larger than the RHS of \eqref{eq:SuffCondSPperpRelaxBinary}.  

\subsection{$(\epsilon, \SetPair_{P,Q})$-Pufferfish Privacy} \label{sec:SuffCondSPQ}

For two users emitting different r.v.s, include one (and excluding the other) will result in different probability distribution in the query answer $f(D)$. 
Theorem below shows that attaining pufferfish privacy on $\SetPair_{P,Q}$ depends on the statistics of these two user.

\begin{theorem} \label{theo:SuffCondSPQ}
	Adding Laplace noise $N_\theta$ with
	\begin{equation}
		 \theta = \frac{1}{\epsilon} \sup_{t \in \D_i \colon i \in \N, \rho } \triangle^*(t)
	\end{equation}
	attains $(\epsilon, \SetPair_{P,Q})$-pufferfish privacy in $Y$, where 
	\begin{align} \label{eq:Tri_t}
		\triangle^*(&t)  = \nonumber \\
		&\begin{cases}
			0 &  F_{P_{i}}(t)  = F_{Q_{i}}(t)  \\
			\min \Set{\triangle>0 \colon  F_{P_i}(t)  \leq  F_{Q_i}(t+\triangle)} & F_{P_i}(t)  > F_{Q_i}(t) \\
			\min \Set{\triangle>0 \colon F_{P_i}(t) \geq  F_{Q_i}(t-\triangle)}  & F_{P_i}(t)  < F_{Q_i}(t)
		\end{cases}.
	\end{align}
	
\end{theorem}
\begin{proof}		
		Since
		\begin{align*}
			P_{Y|S}&(y|s_{P_{i}},\rho) - e^{\epsilon} P_{Y|S}(y|s_{Q_{i}},\rho)  \\
			&=\Big(  \int P_{N_\theta}(y-x) \Pr \big( f(D) = x | D_i \sim P_{i}(\cdot) \big) \dif x - \\
			&\qquad e^{\epsilon} \int P_{N_\theta}(y-x') \Pr \big( f(D) = x | D_i \sim Q_{i}(\cdot) \big) \dif x' \Big) \prod_{j\in\N \colon j \neq i} \zeta_i, 
		\end{align*}
		the problem is equivalent to attaining pufferfish privacy for two prior distributions $ \Pr \big( f(D) = x | D_i \sim P_{i}(\cdot) \big)  $ and $\Pr \big( f(D) = x | D_i \sim Q_{i}(\cdot) \big)$, regardless of the presence probability $\prod_{j\in\N} \zeta_i$. 
		We apply the Kantorovich method to work out a sufficient condition below. 
		
		Define cumulative density functions of two priors using an expression different from~\eqref{eq:P_sDi} without $\zeta_i$ for all $i \in \N$ by 
		\begin{align*}
			F_{X|S}(x|s_{P_i},\rho) &=  \int_{-\infty}^{x}  \int \Pr(f(D_{-i}) = t) P_i(m-t) \dif t  \dif m \\
			&= \int  \Pr(f(D_{-i}) = t)  F_{P_i}(x-t) \dif t   \\
			&= \int \Pr(f(D_{-i}) = x-t)  F_{P_i}(t) \dif t,  \\
			F_{X|S}(x'|s_{Q_i},\rho) &= \Pr(f(D_{-i}) = x'-t)  F_{Q_i}(t) \dif t  , 
		\end{align*}		
		where $ F_{P_i}(t)$ and $ F_{Q_i}(t)$ are the cumulative distribution of $P_{i}(\cdot)$ and $ Q_{i}(\cdot)$, respectively.
		
		We apply Lemma~\ref{lemma:ComputePiStar} here. 
		As $F_{X|S}(x|s_{P_i},\rho)  - F_{X|S}(x|s_{Q_i},\rho) \leq \max_{t \in \D_i} \Set{F_{P_i}(t)  - F_{Q_i}(t) }$, we have $F_{X|S}(x|s_{P_i},\rho)  < F_{X|S}(x|s_{Q_i},\rho) $ if $\max_{t \in \D_i} \Set{F_{P_i}(t)  - F_{Q_i}(t) } < 0 $. Likewise,  $F_{X|S}(x|s_{P_{i}},\rho)  > F_{X|S}(x|s_{Q_i},\rho) $ if $\min_{t \in \D_i} \Set{F_{P_i}(t)  - F_{Q_i}(t) } > 0 $. 
		For given $\triangle$, 
		\begin{align*}
			F_{X|S}(x+\triangle|&s_{P_i},\rho)  - F_{X|S}(x|s_{Q_i},\rho) \\
			& = \int \Pr(f(D_{-i}) = x-t)  \Big(  F_{P_i}(t+\triangle) -  F_{Q_i}(t) \Big) \dif t \\
			&\begin{cases}
				\leq \max_{t \in \D_i} \Set{  F_{P_i}(t+\triangle) -  F_{Q_i}(t) } \\
				\geq \min_{t \in \D_i} \Set{  F_{P_i}(t+\triangle) -  F_{Q_i}(t) } 
			\end{cases}
		\end{align*}
		for all $\triangle > 0$.  Therefore,  \eqref{eq:Tri_t} is tuned to the maximum value over all possible probability distribution $ \Pr(f(D_{-i}) = \cdot) $. Therefore, \eqref{theo:SuffCondSPQ} is sufficient to attain pufferfish privacy in $\SetPair_{P,Q}$. 		
\end{proof}

Theorem~\ref{theo:SuffCondSPQ} and its proof state that the statistical indistinguishability between secrets $s_{P_i}$ and $s_{Q_i}$ depends on the statistics of user $i$ only. That is, the sufficient condition in Theorem~\ref{theo:SuffCondSPQ} applies to any given system containing $-i$, where the $i$the user is joining with different probability features. 

For the existing three users in Table~\ref{tab:ExpSetting}, we add user $4$ based on the two choices in Table~\ref{tab:ExpSettingExtraTwo}, either $D_4 \sim P_4(\cdot)$ or $D_4 \sim Q_4(\cdot)$, corresponds to secret pair $(s_{P_4},s_{Q_4})$. Applying Theorem~\ref{theo:SuffCondSPQ}, we have $\sup_{t \in \D_4 } \triangle^*(t) = 2$ and therefore setting $\theta = \frac{2}{\epsilon}$ attains $\epsilon$-pufferfish privacy on $(s_{P_4},s_{Q_4})$. Note, this calculation does not involve any information of the existing three users, i.e., the result remains $\theta = \frac{2}{\epsilon}$ regardless of how the probabilities in Table~\ref{tab:ExpSetting} changes. 
Alternatively, we can obtain the prior probabilities of query answer $P_{X|S}(\cdot|s_{P_4})$ and $P_{X|S}(\cdot|s_{Q_4})$ to work out the Kantorovich optimal transport plan $\pi^*$ in Figure~\ref{fig:SPQ}, where $\sup_{(x,x') \in \supp(\pi^*)}=2$ and by Proposition~\ref{prop:Kantorovich}, we still get $\theta = \frac{2}{\epsilon}$. But, the computational cost is higher than Theorem~\ref{theo:SuffCondSPQ}, as we need to plug in probabilities in Table~\ref{tab:ExpSetting} to calculate in $P_{X|S}(\cdot|s_{P_4})$ and $P_{X|S}(\cdot|s_{Q_4})$.

\begin{figure}[t]
	\vspace*{-0.8cm}
	\includegraphics[scale=0.6]{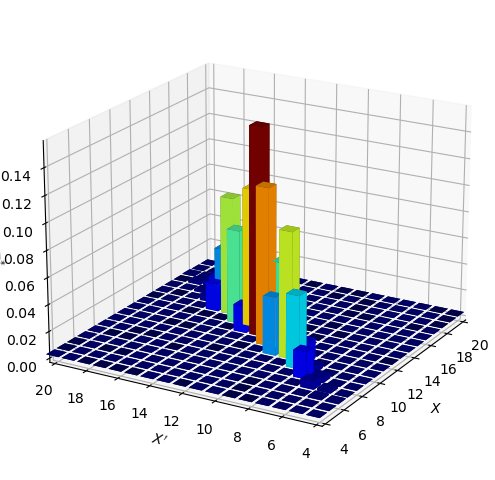}
	\caption{For the first three users in Table~\ref{tab:ExpSetting}, consider adding user $4$ with two choices $D_4 \sim P_4(\cdot)$ and $D_4 \sim Q_4(\cdot)$ in Table~\ref{tab:ExpSettingExtraTwo}. For the secret pair $\SetPair_{P,Q} = \Set{(s_{P_4}, s_{Q_4})}$, the resulting Kantorovich optimal transport plan $\pi^*$ is shown in figure, where $\sup_{(x,x') \in \supp(\pi^*)}=2$.}
	\label{fig:SPQ}
\end{figure}

\subsubsection{Binary $D_i$}

For binary r.v. $D_i$, a more straightforward results than Theorem~\ref{theo:SuffCondSPQ} can be obtained. 

\begin{lemma}  \label{lemma:SuffCondSPQBinary}
	If $P_i$ and $Q_i$ are Bernoulli distributions for all $i \in \N$ and $\rho$, adding Laplace noise  $N_\theta$ with
	\begin{equation}
		\theta = \frac{1}{\epsilon}
	\end{equation}
	attains $(\epsilon, \SetPair_{P,Q})$-pufferfish privacy in $Y$. 
\end{lemma}
\begin{proof}
	Denote $p_i = P_i(1)$ and  $q_i = Q_i(1)$. We prove the lemma by Proposition~\ref{prop:ComputePiStar} below. 
	\begin{align}
		F_{X|S}(x|s_{P_i},&\rho) - F_{X|S}(x|s_{Q_i},\rho) \nonumber  \\
		&= \int \Pr(f(D_{-i}) = x-t) \big( F_{P_i}(t) - F_{Q_i}(t) \big) \dif t  \label{eq:SuffCondSPQBinary}\\
		&= \Pr(f(D_{-i}) = x) \big( P_i(0) - Q_i(0) \big) \nonumber \\
		&= \Pr(f(D_{-i}) = x) \big( q_i - p_i \big)  \nonumber \\
		&\begin{cases}
			 = 0 & p_i = q_i \\
			 > 0 & p_i < q_i \\
			 < 0 & p_i > q_i
		\end{cases} \nonumber 
	\end{align}
	Using \eqref{eq:SuffCondSPQBinary}, 
	\begin{align*}
		F_{X|S}(x|s_{P_i},&\rho) - F_{X|S}(x-1|s_{Q_i},\rho) \nonumber  \\
		&=  \int \Pr(f(D_{-i}) = x-t) \big( F_{P_i}(t) - F_{Q_i}(t-1) \big) \dif t\\
		&=  (1-p_i) \Pr(f(D_{-i}) = x) +  q_i \Pr(f(D_{-i}) = x-1) \geq 0, \\
		 F_{X|S}(x|s_{P_i},&\rho) - F_{X|S}(x+1|s_{Q_i},\rho) \\
		 &=  \int \Pr(f(D_{-i}) = x-t) \big( F_{P_i}(t) - F_{Q_i}(t+1) \big) \dif t\\
		 &= -p_i  \Pr(f(D_{-i}) = x) - (1-q_i)  \Pr(f(D_{-i}) = x+1)  \leq 0.
	\end{align*}
	Therefore, $\max_{\rho, (s_i, s_j) \in \mathbb{S}} \sup_{x \in \X} \triangle^*(x) = 1$ and $\theta = 1/\epsilon$ attains pufferfish privacy. The other convenient way is to apply \autoref{theo:SuffCondSPQ}, which is more straightforward to see that $ \sup_{t \in \D_i \colon i \in \N, \rho } \triangle^*(t) = 1$. 
\end{proof}

\paragraph{Relaxed sufficient condition}

While Lemma~\ref{lemma:SuffCondSPQBinary} is just a simplified version of Theorem~\ref{theo:SuffCondSPQ} for binary $D_i$, we can improve it by a relaxed sufficient condition which produces a smaller value of the scale parameter $\theta$ that reduce the Laplace noise amount for attaining pufferfish privacy. 

Rewrite  \eqref{eq:W1:SuffCond1Main} as
\begin{align*} 
	 \int P_{N_{\theta}}&(y-x') \big( e^{\frac{|x-x'|}{\theta}}  - e^{\epsilon} \big)  \dif \pi^*(x,x')\\
	 & = \int P_{N_{\theta}}(y-x') \underbrace{ \Big( \int (e^{\frac{|x-x'|}{\theta}} - e^{\epsilon}) \pi^*(x,x') \dif x  \Big) }_{\leq0, \text{ for relaxed condition}} \dif x'
	, \  \forall y.
\end{align*}
Then, it is sufficient to request 
\begin{equation} \label{eq:RelaxSuff}
	 \int (e^{\frac{|x-x'|}{\theta}} - e^{\epsilon}) \pi^*(x,x') \dif x  \leq 0, \quad \forall x'
\end{equation}
to make $P_{Y|s_i, \rho} \leq e^{\epsilon} P_{Y|S}(y|s_j,\rho)$. 
This is a relaxed sufficient condition that is proposed in \cite[Theorem~2]{Ding2022AISTATS}, which will result in a smaller value of $\theta$ that attains pufferfish privacy.

\begin{figure}[t]
	\vspace*{-0.8cm}
	\includegraphics[scale=0.6]{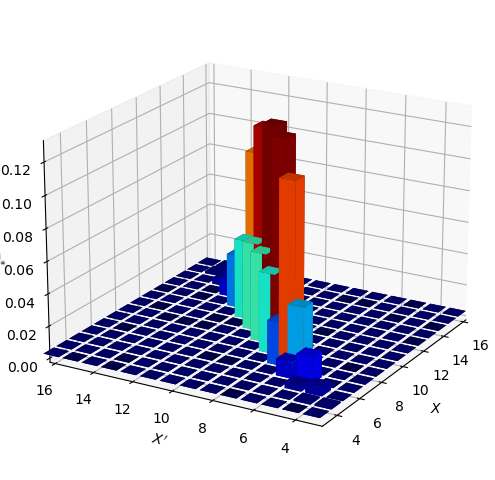}
	\caption{Adding user $4$ to Table~\ref{tab:ExpSetting} by assuming binary r.v. with two choices of probability: $P_4(\cdot)$ and $Q_4(\cdot)$ being $\text{Bernoulli(0.2)}$ and $ \text{Bernoulli(0.9)}$, respectively. The Kantorovich optimal transport plan is shown in figure, where $\sup_{(x,x') \in \supp(\pi^*)}=1$, verifying the sufficient condition $\theta = \frac{1}{\epsilon}$ in Lemma~\ref{lemma:SuffCondSPQBinary}.}
	\label{fig:SuffCondSPQBinary}
\end{figure}

\begin{lemma}[Relaxed sufficient condition]  \label{lemma:SuffCondSPQBinaryRelax}
	If $P_i$ and $Q_i$ are Bernoulli distributions for all $i \in \N$ and $\rho$, adding Laplace noise  $N_\theta$ with
	\begin{equation} \label{eq:SuffCondSPQBinaryRelax}
		\theta = \sup_{x \in \X, \in \N,\rho} \frac{1}{\log \Big( e^{\epsilon} + (e^{\epsilon}-1) \psi_i(x) \Big)}
	\end{equation}
	attains $(\epsilon, \SetPair_{P,Q})$-pufferfish privacy in $Y$, where 
	$$
		\psi_i(x) = 
		\begin{cases}
			\frac{1}{q_i - p_i} \Big( (1-q_i) \frac{\Pr(f(D_{-i}) = x)}{\Pr(f(D_{-i}) = x-1)} + p_i \Big) & p_i < q_i \\
			\frac{1}{p_i - q_i} \Big( (1-p_i) \frac{\Pr(f(D_{-i}) = x)}{\Pr(f(D_{-i}) = x-1)} + q_i \Big) & p_i > q_i
		\end{cases}
	$$
	
\end{lemma}
\begin{proof}
From the proof of  Lemma~\ref{lemma:SuffCondSPQBinary}, it can be seen that if $p_i < q_i$,  
\begin{align} \label{eq:RelaxSuff2}
	\min\Set{ F_{X|S}(x|s_{P_i},\rho), &F_{X|S}(x'|s_{Q_i},\rho) } \nonumber \\
	&= 
	\begin{cases}
		F_{X|S}(x'|s_{Q_i},\rho) & x' \leq x \\
		F_{X|S}(x|s_{P_i},\rho) & x' > x
	\end{cases}
\end{align}
and the support of $\pi^*$ is 
$ \supp(\pi^*) = \Set{(x,x') \colon x = x'-1 \text{ or } x = x' } .$
We have  \eqref{eq:RelaxSuff} being 
\begin{multline} \label{eq:RelaxSuff1}
	\pi^*(x',x') + e^{\frac{1}{\theta}} \pi^*(x'-1,x') \leq (\pi^*(x',x') + \pi^*(x'-1,x')) e^\epsilon \\
	\Longrightarrow e^{\frac{1}{\theta}} \leq e^{\epsilon} + (e^{\epsilon} - 1) \frac{\pi^*(x',x')}{\pi^*(x'-1,x')}
\end{multline}
for all $x' \in \Set{1,\dotsc,|\N|}$. 
Following  \eqref{eq:RelaxSuff2}, \eqref{eq:PiStarCompute1} and  \eqref{eq:PiStarCompute2}, we can calculate $\pi^*(x',x')$ and $\pi^*(x'-1,x')$ as follows.  
\begin{align*}
	\pi^*(x',x') & = F_{X|S}(x'|s_{Q_i},\rho) -  F_{X|S}(x'-1|s_{Q_i},\rho)  \\
	&\qquad\quad -  F_{X|S}(x'-1|s_{P_i},\rho) + F_{X|S}(x'-1|s_{Q_i},\rho) \\
	&= F_{X|S}(x'|s_{Q_i},\rho) -  F_{X|S}(x'-1|s_{P_i},\rho) \\
	&= p_i  \Pr(f(D_{-i}) = x-1) + (1-q_i)  \Pr(f(D_{-i}) = x) , \\
	\pi^*(x'-1,x') & =  F_{X|S}(x'-1|s_{P_i},\rho) -  F_{X|S}(x'-1|s_{Q_i},\rho) \\
	&\qquad\quad -  F_{X|S}(x'-2|s_{P_i},\rho) + F_{X|S}(x'-2|s_{P_i},\rho) \\
	&= F_{X|S}(x'-1|s_{P_i},\rho) -  F_{X|S}(x'-1|s_{Q_i},\rho) \\
	&= \big( q_i - p_i \big) \Pr(f(D_{-i}) = x-1).
\end{align*}
Substitute in~\eqref{eq:RelaxSuff1} and choose the minimum value of $\theta$. We get
$$  \frac{\pi^*(x',x')}{\pi^*(x'-1,x')}  = \frac{1}{q_i - p_i} \Big( (1-q_i) \frac{\Pr(f(D_{-i}) = x)}{\Pr(f(D_{-i}) = x-1)} + p_i \Big) $$
If $p_i > q_i$, 
\begin{align} \label{eq:RelaxSuff3}
	\min\Set{ F_{X|S}(x|s_{P_i},\rho), &F_{X|S}(x'|s_{Q_i},\rho) } \nonumber \\
	&= 
	\begin{cases}
		F_{X|S}(x'|s_{Q_i},\rho) & x' < x \\
		F_{X|S}(x|s_{P_i},\rho) & x' \geq x
	\end{cases}
\end{align}
and the support of $\pi^*$ is 
$ \supp(\pi^*) = \Set{(x,x') \colon x = x'+1 \text{ or } x = x' } .$
for all $x' \in \Set{0,\dotsc,|\N|-1}$. And, 
\begin{align*}
	\pi^*(x',x') & = F_{X|S}(x'|s_{P_i},\rho) -  F_{X|S}(x'-1|s_{Q_i},\rho)  \\
	&\qquad\quad -  F_{X|S}(x'-1|s_{P_i},\rho) + F_{X|S}(x'-1|s_{P_i},\rho) \\
	&= F_{X|S}(x'|s_{P_i},\rho) -  F_{X|S}(x'-1|s_{Q_i},\rho) \\
	&=  (1-p_i) \Pr(f(D_{-i}) = x) +  q_i \Pr(f(D_{-i}) = x-1), \\
	\pi^*(x'+1,x') & =  F_{X|S}(x'|s_{Q_i},\rho) -  F_{X|S}(x'-1|s_{Q_i},\rho) \\
	&\qquad\quad -  F_{X|S}(x'|s_{P_i},\rho) + F_{X|S}(x'-1|s_{Q_i},\rho) \\
	&= F_{X|S}(x'|s_{Q_i},\rho) -  F_{X|S}(x'|s_{P_i},\rho)  \\
	&= \big( p_i - q_i \big) \Pr(f(D_{-i}) = x-1).
\end{align*}
So, the sufficient condition~\eqref{eq:RelaxSuff} is 
\begin{multline*} 
	\pi^*(x',x') + e^{\frac{1}{\theta}} \pi^*(x'+1,x') \leq (\pi^*(x',x') + \pi^*(x'+1,x')) e^\epsilon \\
	\Longrightarrow e^{\frac{1}{\theta}} \leq e^{\epsilon} + (e^{\epsilon} - 1) \frac{\pi^*(x',x')}{\pi^*(x'+1,x')}
\end{multline*}
where 
$$  \frac{\pi^*(x',x')}{\pi^*(x'+1,x')}  = \frac{1}{p_i - q_i} \Big( (1-p_i) \frac{\Pr(f(D_{-i}) = x)}{\Pr(f(D_{-i}) = x-1)} + q_i \Big) $$
Lemma proves. 
\end{proof}

\begin{figure}[t]
	\scalebox{0.6}{
\begin{tikzpicture}

\definecolor{darkgray176}{RGB}{176,176,176}
\definecolor{darkorange25512714}{RGB}{255,127,14}
\definecolor{steelblue31119180}{RGB}{31,119,180}

\begin{axis}[
width=4in,
height=2.5in,
scale only axis,
xmin=0,
xmax=0.04,
xlabel={\Large privacy budget $\epsilon$},
ymin=0,
ymax=1000,
ylabel={\Large $\theta$},
grid=major,
legend style={at={(0.98,0.93)},draw=darkgray!60!black,fill=white,legend cell align=left}
]

\addplot [
color=blue,
line width = 1.5pt]
table {%
0.001 1000
0.00179591836734694 556.818181818182
0.00259183673469388 385.826771653543
0.00338775510204082 295.180722891566
0.00418367346938775 239.024390243902
0.00497959183673469 200.819672131148
0.00577551020408163 173.144876325088
0.00657142857142857 152.173913043478
0.00736734693877551 135.734072022161
0.00816326530612245 122.5
0.00895918367346939 111.617312072893
0.00975510204081633 102.510460251046
0.0105510204081633 94.7775628626693
0.0113469387755102 88.1294964028777
0.0121428571428571 82.3529411764706
0.0129387755102041 77.2870662460568
0.013734693877551 72.8083209509658
0.014530612244898 68.8202247191011
0.0153265306122449 65.2463382157124
0.0161224489795918 62.0253164556962
0.0169183673469388 59.1073582629674
0.0177142857142857 56.4516129032258
0.0185102040816327 54.0242557883131
0.0193061224489796 51.7970401691332
0.0201020408163265 49.746192893401
0.0208979591836735 47.8515625
0.0216938775510204 46.0959548447789
0.0224897959183673 44.464609800363
0.0232857142857143 42.9447852760736
0.0240816326530612 41.5254237288136
0.0248775510204082 40.1968826907301
0.0256734693877551 38.9507154213037
0.026469387755102 37.7794911333847
0.027265306122449 36.6766467065868
0.0280612244897959 35.6363636363636
0.0288571428571429 34.6534653465347
0.0296530612244898 33.7233310392292
0.0304489795918367 32.8418230563003
0.0312448979591837 32.0052253429131
0.0320408163265306 31.2101910828026
0.0328367346938776 30.4536979490367
0.0336326530612245 29.7330097087379
0.0344285714285714 29.045643153527
0.0352244897959184 28.3893395133256
0.0360204081632653 27.7620396600567
0.0368163265306122 27.1618625277162
0.0376122448979592 26.587086272382
0.0384081632653061 26.0361317747078
0.0392040816326531 25.5075481520042
0.04 25
};
\addlegendentry{\Large $\theta$ by Lemma~\ref{lemma:SuffCondSPQBinary}};

\addplot [
color=red,
dashed,
line width = 1.5pt]
table {%
0.001 624.154877161015
0.00179591836734694 347.624008431036
0.00259183673469388 240.930956324606
0.00338775510204082 184.370753739136
0.00418367346938775 149.330993903766
0.00497959183673469 125.492448146978
0.00577551020408163 108.22422586348
0.00657142857142857 95.1389735113898
0.00736734693877551 84.8809917109327
0.00816326530612245 76.623303840413
0.00895918367346939 69.8328019637316
0.00975510204081633 64.1503631615613
0.0105510204081633 59.3252265559775
0.0113469387755102 55.1769886376612
0.0121428571428571 51.5725449686247
0.0129387755102041 48.4115423605899
0.013734693877551 45.6168892855455
0.014530612244898 43.128385055871
0.0153265306122449 40.8983340535302
0.0161224489795918 38.8884590356808
0.0169183673469388 37.0676857146936
0.0177142857142857 35.410524582732
0.0185102040816327 33.8958702148819
0.0193061224489796 32.5060975734127
0.0201020408163265 31.2263729971174
0.0208979591836735 30.0441226403522
0.0216938775510204 28.9486179254975
0.0224897959183673 27.9306490209393
0.0232857142857143 26.98226528325
0.0240816326530612 26.0965671710051
0.0248775510204082 25.2675381029541
0.0256734693877551 24.4899075921741
0.026469387755102 23.7590390730596
0.027265306122449 23.0708373753828
0.0280612244897959 22.4216719445845
0.0288571428571429 21.8083127681857
0.0296530612244898 21.2278766209994
0.0304489795918367 20.677781741055
0.0312448979591837 20.1557094329119
0.0320408163265306 19.6595713937902
0.0328367346938776 19.1874817915278
0.0336326530612245 18.7377333071946
0.0344285714285714 18.3087765007836
0.0352244897959184 17.8992019743678
0.0360204081632653 17.5077249000295
0.0368163265306122 17.1331715546986
0.0376122448979592 16.7744675646206
0.0384081632653061 16.4306276114581
0.0392040816326531 16.1007463923089
0.04 15.7839906589862
};

\addlegendentry{\Large $\theta$ by relaxed sufficient condition in Lemma~\ref{lemma:SuffCondSPQBinaryRelax}};
\end{axis}

\end{tikzpicture}}
	\caption{For the same setting in Figure~\ref{fig:SuffCondSPQBinary}, we apply Lemma~\ref{lemma:SuffCondSPQBinaryRelax} instead to get the value of $\theta$ in \eqref{eq:SuffCondSPQBinaryRelax}. It is smaller than $\theta= \frac{1}{\epsilon}$ in Lemma~\ref{lemma:SuffCondSPQBinary} for all value of $\epsilon$.}
	\label{fig:SuffCondSPQBinaryRelax}
\end{figure}
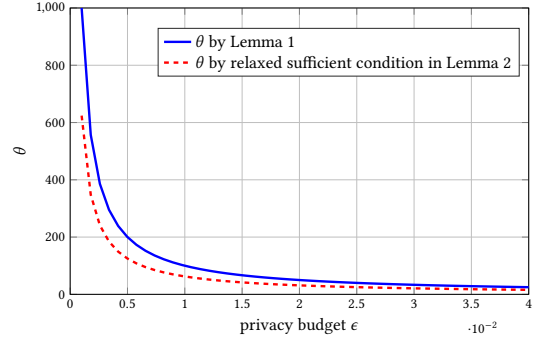

Consider adding a $4$th user obtaining a binary r.v. to the system in Table~\ref{tab:ExpSetting} over two choices $D_4 \sim \text{Bernoulli}(0.2)$ and $D_4 \sim \text{Bernoulli}(0.9)$. Following Lemma~\ref{lemma:SuffCondSPQBinaryRelax}, without any calculation, we need to set $\theta = \frac{1}{\epsilon}$. This can be seem by the resulting Kantorovich optimal transport plan $\pi^*$ in Figure~\ref{fig:SuffCondSPQBinary}, where $\sup_{(x,x') \in \supp(\pi^*)}=1$. 
Instead, applying Lemma~\ref{lemma:SuffCondSPQBinaryRelax} by specifying $p_4= 0.2$ and $q_4= 0.9$, we plot the resulting $\theta$ in \eqref{eq:SuffCondSPQBinaryRelax} and compare to $\frac{1}{\epsilon}$ in Figure~\ref{fig:SuffCondSPQBinaryRelax}.  It can be seen that a smaller value of scale parameter $\theta$ can be produced by Lemma~\ref{lemma:SuffCondSPQBinaryRelax}. 
It should be noted that the computation in \eqref{eq:SuffCondSPQBinaryRelax} requires knowledge of statistics of the rest of users and Bernoulli probability $p_i$ and $q_i$, while $\theta = \frac{1}{\epsilon}$ generally applies regardless of any system parameters.

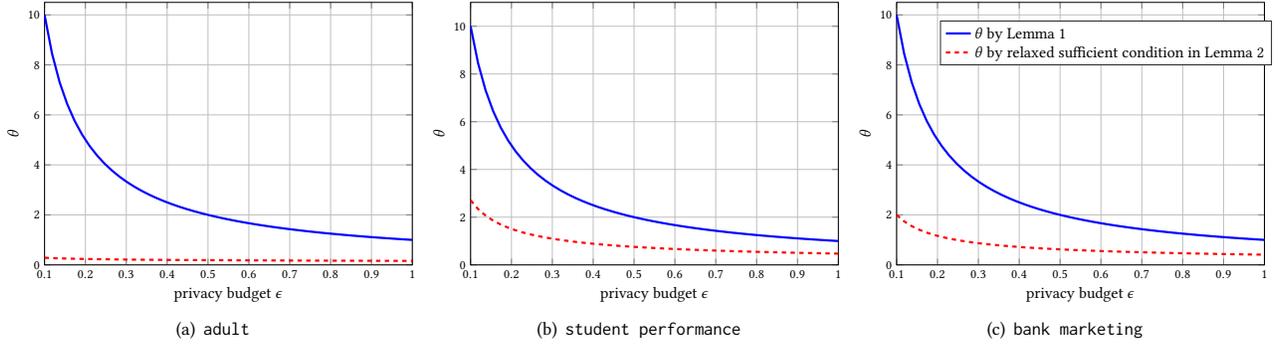
\begin{figure*}[t]
	\subfigure[\texttt{adult}]{\scalebox{0.55}{
\begin{tikzpicture}

\definecolor{darkgray176}{RGB}{176,176,176}
\definecolor{darkorange25512714}{RGB}{255,127,14}
\definecolor{steelblue31119180}{RGB}{31,119,180}

\begin{axis}[
width=3.5in,
height=2.5in,
scale only axis,
xmin=0.1,
xmax=1,
xlabel={\Large privacy budget $\epsilon$},
ymin=0,
ymax=10.5,
ylabel={\Large $\theta$},
grid=major,
legend style={at={(1.02,0.93)},draw=darkgray!60!black,fill=white,legend cell align=left}
]
\addplot [
color=blue,
line width = 1.5pt
]
table {%
0.1 10
0.118367346938776 8.44827586206896
0.136734693877551 7.3134328358209
0.155102040816327 6.44736842105263
0.173469387755102 5.76470588235294
0.191836734693878 5.21276595744681
0.210204081632653 4.75728155339806
0.228571428571429 4.375
0.246938775510204 4.0495867768595
0.26530612244898 3.76923076923077
0.283673469387755 3.52517985611511
0.302040816326531 3.31081081081081
0.320408163265306 3.12101910828025
0.338775510204082 2.95180722891566
0.357142857142857 2.8
0.375510204081633 2.66304347826087
0.393877551020408 2.53886010362694
0.412244897959184 2.42574257425743
0.430612244897959 2.32227488151659
0.448979591836735 2.22727272727273
0.46734693877551 2.13973799126638
0.485714285714286 2.05882352941176
0.504081632653061 1.98380566801619
0.522448979591837 1.9140625
0.540816326530612 1.84905660377358
0.559183673469388 1.78832116788321
0.577551020408163 1.73144876325088
0.595918367346939 1.67808219178082
0.614285714285714 1.62790697674419
0.63265306122449 1.58064516129032
0.651020408163265 1.53605015673981
0.669387755102041 1.49390243902439
0.687755102040816 1.4540059347181
0.706122448979592 1.41618497109827
0.724489795918367 1.38028169014085
0.742857142857143 1.34615384615385
0.761224489795918 1.31367292225201
0.779591836734694 1.28272251308901
0.797959183673469 1.25319693094629
0.816326530612245 1.225
0.83469387755102 1.19804400977995
0.853061224489796 1.17224880382775
0.871428571428571 1.14754098360656
0.889795918367347 1.12385321100917
0.908163265306122 1.10112359550562
0.926530612244898 1.07929515418502
0.944897959183673 1.05831533477322
0.963265306122449 1.03813559322034
0.981632653061225 1.01871101871102
1 1
};

\addplot [
color=red,
dashed,
line width = 1.5pt]
table {%
0.1 0.285742085742714
0.118367346938776 0.27227955597568
0.136734693877551 0.261595245601818
0.155102040816327 0.25282065426752
0.173469387755102 0.245426780867657
0.191836734693878 0.239070443359056
0.210204081632653 0.23351787679279
0.228571428571429 0.228603531292134
0.246938775510204 0.224206382862942
0.26530612244898 0.220235588589964
0.283673469387755 0.216621419848412
0.302040816326531 0.213309321129236
0.320408163265306 0.210255895318646
0.338775510204082 0.207426117697148
0.357142857142857 0.204791357238603
0.375510204081633 0.202327942284721
0.393877551020408 0.200016101816691
0.412244897959184 0.197839171214847
0.430612244897959 0.195782987699292
0.448979591836735 0.193835424059639
0.46734693877551 0.191986024719731
0.485714285714286 0.190225718564719
0.504081632653061 0.188546590066181
0.522448979591837 0.186941695188438
0.540816326530612 0.185404912055149
0.559183673469388 0.183930818859812
0.577551020408163 0.182514593321269
0.595918367346939 0.181151929319831
0.614285714285714 0.179838967340412
0.63265306122449 0.178572236092165
0.651020408163265 0.177348603236839
0.669387755102041 0.176165233587985
0.687755102040816 0.175019553474404
0.706122448979592 0.173909220218457
0.724489795918367 0.172832095881087
0.742857142857143 0.171786224583967
0.761224489795918 0.170769812844885
0.779591836734694 0.169781212462844
0.797959183673469 0.168818905569877
0.816326530612245 0.167881491531596
0.83469387755102 0.166967675431279
0.853061224489796 0.166076257915322
0.871428571428571 0.165206126213192
0.889795918367347 0.164356246174045
0.908163265306122 0.163525655186228
0.926530612244898 0.162713455865821
0.944897959183673 0.161918810417049
0.963265306122449 0.161140935581303
0.981632653061225 0.160379098103227
1 0.159632610652221
};
\end{axis}

\end{tikzpicture}}}
	\subfigure[\texttt{student performance}]{\scalebox{0.55}{
\begin{tikzpicture}

\definecolor{darkgray176}{RGB}{176,176,176}
\definecolor{darkorange25512714}{RGB}{255,127,14}
\definecolor{steelblue31119180}{RGB}{31,119,180}

\begin{axis}[
width=3.5in,
height=2.5in,
scale only axis,
xmin=0.1,
xmax=1,
xlabel={\Large privacy budget $\epsilon$},
ymin=0,
ymax=11,
ylabel={\Large $\theta$},
grid=major,
legend style={at={(1.02,0.93)},draw=darkgray!60!black,fill=white,legend cell align=left}
]
\addplot [
color=blue,
line width = 1.5pt
]
table {%
0.1 10
0.118367346938776 8.44827586206896
0.136734693877551 7.3134328358209
0.155102040816327 6.44736842105263
0.173469387755102 5.76470588235294
0.191836734693878 5.21276595744681
0.210204081632653 4.75728155339806
0.228571428571429 4.375
0.246938775510204 4.0495867768595
0.26530612244898 3.76923076923077
0.283673469387755 3.52517985611511
0.302040816326531 3.31081081081081
0.320408163265306 3.12101910828025
0.338775510204082 2.95180722891566
0.357142857142857 2.8
0.375510204081633 2.66304347826087
0.393877551020408 2.53886010362694
0.412244897959184 2.42574257425743
0.430612244897959 2.32227488151659
0.448979591836735 2.22727272727273
0.46734693877551 2.13973799126638
0.485714285714286 2.05882352941176
0.504081632653061 1.98380566801619
0.522448979591837 1.9140625
0.540816326530612 1.84905660377358
0.559183673469388 1.78832116788321
0.577551020408163 1.73144876325088
0.595918367346939 1.67808219178082
0.614285714285714 1.62790697674419
0.63265306122449 1.58064516129032
0.651020408163265 1.53605015673981
0.669387755102041 1.49390243902439
0.687755102040816 1.4540059347181
0.706122448979592 1.41618497109827
0.724489795918367 1.38028169014085
0.742857142857143 1.34615384615385
0.761224489795918 1.31367292225201
0.779591836734694 1.28272251308901
0.797959183673469 1.25319693094629
0.816326530612245 1.225
0.83469387755102 1.19804400977995
0.853061224489796 1.17224880382775
0.871428571428571 1.14754098360656
0.889795918367347 1.12385321100917
0.908163265306122 1.10112359550562
0.926530612244898 1.07929515418502
0.944897959183673 1.05831533477322
0.963265306122449 1.03813559322034
0.981632653061225 1.01871101871102
1 1
};

\addplot [
color=red,
dashed,
line width = 1.5pt]
table {%
0.1 2.70750129924749
0.118367346938776 2.33776444136251
0.136734693877551 2.06635472773713
0.155102040816327 1.85839316868999
0.173469387755102 1.69376808614489
0.191836734693878 1.56006630150161
0.210204081632653 1.44921019749489
0.228571428571429 1.35571612027451
0.246938775510204 1.27572996715259
0.26530612244898 1.20646333413532
0.283673469387755 1.14584868200559
0.302040816326531 1.09232029557105
0.320408163265306 1.04467056265628
0.338775510204082 1.00195299072149
0.357142857142857 0.963415137916189
0.375510204081633 0.928451217938119
0.393877551020408 0.896567958087157
0.412244897959184 0.867359578196158
0.430612244897959 0.840489168016229
0.448979591836735 0.815674631466135
0.46734693877551 0.792677941916132
0.485714285714286 0.771296832504917
0.504081632653061 0.75135830077276
0.522448979591837 0.732713481416459
0.540816326530612 0.715233562158743
0.559183673469388 0.698806503093456
0.577551020408163 0.683334380808018
0.595918367346939 0.66873122262454
0.614285714285714 0.654921228493145
0.63265306122449 0.641837301854469
0.651020408163265 0.629419828534986
0.669387755102041 0.617615656104675
0.687755102040816 0.606377236281383
0.706122448979592 0.595661900744951
0.724489795918367 0.58543124672902
0.742857142857143 0.575650613427857
0.761224489795918 0.566288633911619
0.779591836734694 0.557316850124875
0.797959183673469 0.548709380828147
0.816326530612245 0.540442634164859
0.83469387755102 0.532495057998073
0.853061224489796 0.52484692234036
0.871428571428571 0.517480129155683
0.889795918367347 0.510378045590389
0.908163265306122 0.503525357327145
0.926530612244898 0.496907939278908
0.944897959183673 0.490512741271868
0.963265306122449 0.484327686724129
0.981632653061225 0.478341582624572
1 0.472544039364862
};
\end{axis}

\end{tikzpicture}}}
	\subfigure[\texttt{bank marketing}]{\scalebox{0.55}{
\begin{tikzpicture}

\definecolor{darkgray176}{RGB}{176,176,176}
\definecolor{darkorange25512714}{RGB}{255,127,14}
\definecolor{steelblue31119180}{RGB}{31,119,180}

\begin{axis}[
width=3.5in,
height=2.5in,
scale only axis,
xmin=0.1,
xmax=1,
xlabel={\Large privacy budget $\epsilon$},
ymin=0,
ymax=10.5,
ylabel={\Large $\theta$},
grid=major,
legend style={at={(1.02,0.93)},draw=darkgray!60!black,fill=white,legend cell align=left}
]

\addplot [
color=blue,
line width = 1.5pt
]
table {%
0.1 10
0.118367346938776 8.44827586206896
0.136734693877551 7.3134328358209
0.155102040816327 6.44736842105263
0.173469387755102 5.76470588235294
0.191836734693878 5.21276595744681
0.210204081632653 4.75728155339806
0.228571428571429 4.375
0.246938775510204 4.0495867768595
0.26530612244898 3.76923076923077
0.283673469387755 3.52517985611511
0.302040816326531 3.31081081081081
0.320408163265306 3.12101910828025
0.338775510204082 2.95180722891566
0.357142857142857 2.8
0.375510204081633 2.66304347826087
0.393877551020408 2.53886010362694
0.412244897959184 2.42574257425743
0.430612244897959 2.32227488151659
0.448979591836735 2.22727272727273
0.46734693877551 2.13973799126638
0.485714285714286 2.05882352941176
0.504081632653061 1.98380566801619
0.522448979591837 1.9140625
0.540816326530612 1.84905660377358
0.559183673469388 1.78832116788321
0.577551020408163 1.73144876325088
0.595918367346939 1.67808219178082
0.614285714285714 1.62790697674419
0.63265306122449 1.58064516129032
0.651020408163265 1.53605015673981
0.669387755102041 1.49390243902439
0.687755102040816 1.4540059347181
0.706122448979592 1.41618497109827
0.724489795918367 1.38028169014085
0.742857142857143 1.34615384615385
0.761224489795918 1.31367292225201
0.779591836734694 1.28272251308901
0.797959183673469 1.25319693094629
0.816326530612245 1.225
0.83469387755102 1.19804400977995
0.853061224489796 1.17224880382775
0.871428571428571 1.14754098360656
0.889795918367347 1.12385321100917
0.908163265306122 1.10112359550562
0.926530612244898 1.07929515418502
0.944897959183673 1.05831533477322
0.963265306122449 1.03813559322034
0.981632653061225 1.01871101871102
1 1
};
\addlegendentry{\Large $\theta$ by Lemma~\ref{lemma:SuffCondSPQBinary}};

\addplot [
color=red,
dashed,
line width = 1.5pt]
table {%
0.1 1.99895593382343
0.118367346938776 1.74173598913922
0.136734693877551 1.55239599341147
0.155102040816327 1.40690414058039
0.173469387755102 1.29139681739232
0.191836734693878 1.19731248528637
0.210204081632653 1.11907639942985
0.228571428571429 1.05290139816288
0.246938775510204 0.996123720120083
0.26530612244898 0.946814601058905
0.283673469387755 0.903542680423838
0.302040816326531 0.865223038130963
0.320408163265306 0.831018110836497
0.338775510204082 0.800270806915549
0.357142857142857 0.77245823476801
0.375510204081633 0.747158991177693
0.393877551020408 0.724029586791379
0.412244897959184 0.702787161598644
0.430612244897959 0.683196614325065
0.448979591836735 0.665060883274554
0.46734693877551 0.64821351281165
0.485714285714286 0.632512901390381
0.504081632653061 0.617837802968053
0.522448979591837 0.6040837739393
0.540816326530612 0.591160341276072
0.559183673469388 0.5789887264281
0.577551020408163 0.567500001571389
0.595918367346939 0.556633585175847
0.614285714285714 0.546336006078479
0.63265306122449 0.53655988166558
0.651020408163265 0.527263068020641
0.669387755102041 0.518407949126005
0.687755102040816 0.509960839221832
0.706122448979592 0.501891477801654
0.724489795918367 0.494172600874884
0.742857142857143 0.48677957535561
0.761224489795918 0.479690085966115
0.779591836734694 0.472883866037403
0.797959183673469 0.466342465170726
0.816326530612245 0.460049047986192
0.83469387755102 0.453988219197296
0.853061224489796 0.448145871067195
0.871428571428571 0.442509049964949
0.889795918367347 0.437065839279657
0.908163265306122 0.431805256392127
0.926530612244898 0.426717161766903
0.944897959183673 0.421792178527288
0.963265306122449 0.41702162112454
0.981632653061225 0.412397431919273
1 0.407912124665829
};
\addlegendentry{\Large $\theta$ by relaxed sufficient condition in Lemma~\ref{lemma:SuffCondSPQBinaryRelax}};
\end{axis}

\end{tikzpicture}}}
	\caption{Two values of $\theta$ produced by  Lemma~\ref{lemma:SuffCondSPQBinary}  and Lemma~\ref{lemma:SuffCondSPQBinaryRelax} for achieving $\epsilon$-pufferfish privacy on $\SetPair_{P,Q} = \Set{(P_i,Q_i)}$ on three datasets in the UCI machine learning repository. 
	For \texttt{adult}, $P_i (\cdot) = \Pr(\texttt{relationship} = \cdot | \texttt{race} = \texttt{`Asian-Pac-Islander'})$ and $Q_i (\cdot) = \Pr(\texttt{relationship} = \cdot | \texttt{race} = \texttt{`Amer-Indian-Eskimo'})$; 
	For \texttt{student performance},  $P_i (\cdot) = \Pr(\texttt{romance} = \cdot | \texttt{freetime} = \texttt{5})$ and $Q_i (\cdot) = \Pr(\texttt{romance} = \cdot | \texttt{freetime} = \texttt{2})$; 
	For \texttt{bank marketing}, $P_i (\cdot) = \Pr(\texttt{loan} = \cdot | \texttt{job} = \texttt{`blue-collar'})$ and $Q_i (\cdot) = \Pr(\texttt{loan} = \cdot | \texttt{job} = \texttt{`management'})$. 
	}
	\label{fig:SuffCondUCIPQBinary}
\end{figure*}

The noise reduction by Lemma~\ref{lemma:SuffCondSPQBinaryRelax} in Figure~\ref{fig:SuffCondSPQBinaryRelax} may not seem very large, because the difference between $p_4= 0.2$ and $q_4= 0.9$. We can get better results in real-world applications.  
See Figure~\ref{fig:SuffCondUCIPQBinary}. We extract binary features \texttt{relationship} (\texttt{married}/\texttt{unmarried}), \texttt{romance} and \texttt{loan} in \texttt{adult}, \texttt{student performance} and \texttt{bank marketing} datasets, respectively. It can be seen that the noise amount is significantly reduced in smaller $\epsilon$ range, which indicates better data utility can be achieved for strict data privacy guarantee.

\section{Conclusion}

We studied how to attain individual level indistinguishability using pufferfish privacy framework. We formulated a multi-user system assuming each user reports a realization of a random variable to answer a summation query. Adopting Laplace noise and Kantorovich mechanism, we considered four secret pair sets: $\SetPair_{a,b}$ and $\SetPair_{a,\perp}$ discriminates different realizations of users and whether or not users are present in the system; $\SetPair_{P,\perp}$ and $\SetPair_{P,Q}$ discriminates different probability distributions of users data and whether or not users are present in the system. 
We derived sufficient condition for attaining pufferfish privacy on each secret pair. It is shown that pufferfish privacy guarantees indistinguishability of individual's data conditions on this individual only, i.e., it is independent of the rest users in the system.  
For pufferfish privacy on $\SetPair_{P,\perp}$ and $\SetPair_{P,Q}$, we can utilize the intrinsic statistics to reduce the noise.  
We run experiments to show that pufferfish privacy on $\SetPair_{P,\perp}$ and $\SetPair_{P,Q}$ guarantees indistinguishability when a class of users are modified, added or removed.

While the multi-user system assigns each user a presence probability indicating the chances of his/her participation, the sufficient conditions derived in this paper are independent of it. The question is under what circumstances (i.e., on which secret pairs), the achievability of pufferfish privacy relies on this presence probability. 
On the other hand, it is worth discussing how to extend the results derived in this paper to R\'{e}nyi pufferfish privacy framework and study how the sufficient conditions change with the R\'{e}nyi parameter $\alpha$. 
In addition, it would be interesting to see how to calibrate Gaussian noise, with fast decaying tail probability, for attaining individual level pufferfish privacy.



\vspace{10pt}

{\noindent \Large \textbf{APPENDIX}}

\appendix
\section{Computing $\sup_{(x,x') \in \supp(\pi^*)} |x-x'|$} \label{app1}

To find the maximum pairwise distance in the support of Kantorovich optimal transport plan $\pi^*$, we do not need to determine the probability mass $\pi^*(x,x')$ for each $(x,x') \in \X^2$. The following proposition shows how to calculate it directly from the prior distributions. 
\begin{lemma}  \label{lemma:ComputePiStar}
	For $F_{X|S}(\cdot|s_i,\rho)$ and  $F_{X|S}(\cdot|s_j,\rho)$ being the cumulative distribution of the priors $P_{X|S}(\cdot|s_i,\rho)$ and  $P_{X|S}(\cdot|s_j,\rho)$, respectively, 
		$$ \sup_{(x,x') \in \supp(\pi^*)}  |x-x'| = \sup_{x \in \X} \triangle^*(x) $$
		where 
		\begin{align}
			\triangle^*(&x)  = \nonumber \\
			&\begin{cases}
				0 &  F_{X|S}(x|s_i,\rho) = F_{X|S}(x|s_j,\rho) \\
				\min \Set{\triangle>0 \colon  F_{X|S}(x|s_i,\rho) \leq \\
					\qquad\qquad F_{X|S}(x+\triangle|s_j,\rho)} & F_{X|S}(x|s_i,\rho) > F_{X|S}(x|s_j,\rho) \\
				\min \Set{\triangle>0 \colon  F_{X|S}(x|s_i,\rho) \geq  \\
					\qquad\qquad  F_{X|S}(x-\triangle|s_j,\rho)} & F_{X|S}(x|s_i,\rho) < F_{X|S}(x|s_j,\rho) 
			\end{cases}
		\end{align}
\end{lemma}
\begin{proof}
	For $(x,x') \in \X^2$, the joint cumulative distribution of the Kantorovich optimal transport plan $\pi^*$ is  
	$$\pi^*((-\infty,x], (-\infty,x']) = \min\Set{F_{X|S}(x|s_i,\rho) , F_{X|S}(x'|s_j,\rho) }.$$ 
	For small $\dif x, \dif  x' > 0 $, consider the numerator of 
	\begin{equation} \label{eq:PiStarCompute1}
		 \pi^*(x,x') = \frac{\dif^2}{\dif x  \dif x'} \pi^*((-\infty,x], (-\infty,x']). 
	\end{equation}
	We have 
	\begin{align}
		&\dif^2 \pi^*((-\infty, x],(-\infty, x']) \nonumber \\
		& =\pi^*((-\infty, x + \dif x], (-\infty, x' + \dif x']) - \pi^*((-\infty, x + \dif {x}], (-\infty, x']) \nonumber \\
		& \quad - \pi^*((-\infty, x], (-\infty, x' + \dif x']) +  \pi^*((-\infty, x], (-\infty, x'])  \nonumber \\
		& =  \min\Set{F_{X|S}(x + \dif {x}|s_i,\rho) , F_{X|S}(x' + \dif {x'}|s_j,\rho ) } \nonumber \\
		& \qquad\qquad - \min\Set{F_{X|S}(x + \dif {x}|s_i,\rho) , F_{X|S}(x'|s_j),\rho } \nonumber  \\
		& \qquad\qquad - \min\Set{F_{X|S}(x|s_i,\rho) , F_{X|S}(x' + \dif {x'}|s_j,\rho) } \nonumber \\
		&\qquad\qquad +  \min\Set{F_{X|S}(x|s_i,\rho) , F_{X|S}(x'|s_j,\rho) }.  \label{eq:PiStarCompute2}
	\end{align}
	It is clear that to know $\supp(\pi^*)$, we just need to judge whether or not $\dif^2 \pi^*((-\infty, x],(-\infty, x']) = 0$. We consider the following cases. 	
	\begin{itemize}
		\item 	If $F_{X|S}(x|s_,\rho) = F_{X|S}(x'|s_j,\rho)$, then 
					\begin{multline*}
						\min\Set{F_{X|S}(x|s_i,\rho) , F_{X|S}(x' + \dif {x'}|s_j,\rho) } = \\
						\min\Set{F_{X|S}(x|s_i,\rho) , F_{X|S}(x'|s_j,\rho) } 
					\end{multline*}
					and so 
					\begin{align*}
						\dif^2 \pi^*((-\infty, x],&(-\infty, x']) \\
						& =  \min\Set{F_{X|S}(x + \dif {x}|s_i,\rho) , F_{X|S}(x' + \dif {x'}|s_j,rho) } \\
						& \qquad - \min\Set{F_{X|S}(x + \dif {x}|s_i,\rho) , F_{X|S}(x'|s_j,\rho)} \geq 0.
					\end{align*}
		\item 	If  $F_{X|S}(x|s_i,\rho) > F_{X|S}(x'|s_j,\rho)$, then 
					\begin{multline*}
						\min\Set{F_{X|S}(x + \dif {x}|s_i,\rho) , F_{X|S}(x'|s_j,\rho) }  = \\
						\min\Set{F_{X|S}(x|s_i,\rho) , F_{X|S}(x'|s_j,\rho) }
					\end{multline*}
					and 
					\begin{align*}
						\dif^2 & \pi^*((-\infty, x],(-\infty, x']) \\
						&= \min\Set{F_{X|S}(x + \dif {x}|s_i,\rho) , F_{X|S}(x' + \dif {x'}|s_j,\rho) } - \\
						&\qquad\qquad   \min\Set{F_{X|S}(x|s_i,\rho) , F_{X|S}(x' + \dif {x'}|s_j,\rho) } \\
						&=
						\begin{cases}
							0	& F_{X|S}(x|s_i,\rho) \geq  F_{X|S}(x' + \dif {x'} |s_j,\rho) \\
							\min\Set{F_{X|S}(x + \dif {x}|s_i,\rho) , \\ 
								\quad F_{X|S}(x' + \dif {x'}|s_j,\rho) }  \\
							\quad - F_{X|S}(x|s_i,\rho)  \geq 0 & F_{X|S}(x|s_i,\rho) <  F_{X|S}(x' + \dif {x'} |s_j,\rho) 
						\end{cases}
					\end{align*}
		\item If  $F_{X|S}(x|s_i,\rho) < F_{X|S}(x'|s_j,\rho)$, then 
					\begin{multline*}
						\min\Set{F_{X|S}(x|s_i,\rho) , F_{X|S}(x' + \dif {x'}|s_j,\rho) } = \\
						\min\Set{F_{X|S}(x|s_i,\rho) , F_{X|S}(x'|s_j,\rho) }
					\end{multline*}
					and 
					\begin{align*}
						& \dif^2 \pi^*((-\infty, x],(-\infty, x']) \\
						& =  \min\Set{F_{X|S}(x + \dif {x}|s_i,\rho) , F_{X|S}(x' + \dif {x'}|s_j,\rho)} \\
						& \qquad - \min\Set{F_{X|S}(x + \dif {x}|s_i,\rho) , F_{X|S}(x'|s_j,\rho)} \\
						&= 
						\begin{cases}
							0 & F_{X|S}(x+\dif x|s_i,\rho) \leq F_{X|S}(x'|s_j,\rho)\\
							\min\Set{F_{X|S}(x + \dif {x}|s_i,\rho) , \\
								\quad F_{X|S}(x' + \dif {x'}|s_j,\rho)} \\
							\quad  - F_{X|S}(x'|s_i,\rho) \geq 0 & F_{X|S}(x+\dif x|s_i,\rho) > F_{X|S}(x'|s_j,\rho)
						\end{cases}
					\end{align*}	
	\end{itemize}
	We then characterize $\supp(\pi^*)$ by considering nonzero mass in terms of the pairwise distance as follows. Let $x = x'$, consider three cases below. 		
	\begin{enumerate}
		\item When $F_{X|S}(x|s_i,\rho) = F_{X|S}(x|s_j,\rho)$, then $\pi^*(x,x) \geq 0$. But, for all $\triangle > 0$, $F_{X|S}(x|s_i,\rho) \leq F_{X|S}(x+\triangle|s_j,\rho)$and $F_{X|S}(x|s_i,\rho) \geq F_{X|S}(x-\triangle|s_j,\rho)$ and therefore $\pi^*(x,x+\triangle) = 0$ and $\pi^*(x,x-\triangle) = 0$. 
		\item When $F_{X|S}(x|s_i,\rho) > F_{X|S}(x|s_j,\rho)$, define $\triangle^*(x) = \min \Set{\triangle>0 \colon  F_{X|S}(x|s_i,\rho) \leq  F_{X|S}(x+\triangle|s_j,\rho)}$. Then, $F_{X|S}(x|s_i,\rho) \leq F_{X|S}(x+\triangle|s_j,\rho)$ for all $\triangle > \triangle^*(x)$ and $F_{X|S}(x|s_i,\rho) > F_{X|S}(x-\triangle|s_j,\rho)$ for all $\triangle>0$. That is, $\pi^*(x,x+\triangle) = 0$ for all $\triangle > \triangle^*(x)$ and $\pi^*(x,x-\triangle) = 0$ for all $\triangle>0$. 
		\item When $F_{X|S}(x|s_i,\rho) < F_{X|S}(x|s_j,\rho)$, define $\triangle^*(x) = \min \Set{\triangle>0 \colon  F_{X|S}(x|s_i,\rho) \geq  F_{X|S}(x-\triangle|s_j,\rho)}$. Then, $F_{X|S}(x|s_i,\rho) < F_{X|S}(x+\triangle|s_j,\rho)$ for all $\triangle > 0$ and $F_{X|S}(x|s_i,\rho) \geq F_{X|S}(x-\triangle|s_j,\rho)$ for all $\triangle>\triangle^*(x)$. That is, $\pi^*(x,x+\triangle) = 0$ for all $\triangle > 0$ and $\pi^*(x,x-\triangle) = 0$ for all $\triangle > \triangle^*(x)$. 
	\end{enumerate}	
	Therefore, for case~(1), $\sup_{x' \in \X}  |x-x'| = 0$, for cases~(2) and (3), $\sup_{x' \in \X}  |x-x'| = \triangle^*(x)$. This proves the proposition. 
\end{proof}

Recall \eqref{eq:piStar}, that $\pi^*(x,x') = \frac{\dif^2}{\dif x  \dif x'} \min \big\{ F_{X|S}(x|s_i,\rho), F_{X|S}(x'|s_j,\rho) \big\}$. 
An intuitive interpretation of Lemma~\ref{lemma:ComputePiStar} is that the Kantorovich optimal transport plan will assign nonzero mass to $(x,x')$ where the minimum changes from $F_{X|S}(x|s_i,\rho),$ to $F_{X|S}(x|s_j,\rho),$ or vice versa. With Lemma~\ref{lemma:ComputePiStar}, it is clear that we can revise the Kantorovich mechanism in Proposition~\ref{prop:Kantorovich} as follows. 
\begin{proposition}  \label{prop:ComputePiStar}
	Adding Laplace noise $N_\theta$ with 	
	\begin{equation}\label{eq:w_1_mechanism_alt}
		\theta = \frac{1}{\epsilon}\max_{\rho, (s_i, s_j) \in \mathbb{S}} \sup_{x \in \X} \triangle^*(x)
	\end{equation}
	attains $(\epsilon,\SetPair)$-pufferfish privacy in $Y$. \qed
\end{proposition}

\begin{example}
	Consider two prior distributions in \autoref{tab:ex2P}. We have the corresponding cumulative distributions in \autoref{tab:ex2F}. 
	\begin{table}[t]
		\caption{Examples of $P_{X|S}$} \label{tab:ex2P}
		\begin{center}
			\begin{tabular}{cccccc}
				\hline\hline
				& $X = 1$ & $X=2$ & $X=3$ & $X=4$ & $X=5$\\ \hline
				$P_{X|S}(\cdot|s_i,\rho)$ & 0.2  & 0.225 & 0.5 & 0.075 & 0 \\ \hline
				$P_{X|S}(\cdot|s_j,\rho)$ & 0  & 0.075 & 0.5 & 0.225 & 0.2 \\  \hline
			\end{tabular}
		\end{center}
	\end{table}
		\begin{table}[t]
		\caption{Cumulative Probability Distribution of Table~\ref{tab:ex2P}} \label{tab:ex2F}
		\begin{center}
			\begin{tabular}{cccccc}
				\hline\hline
				& $X = 1$ & $X=2$ & $X=3$ & $X=4$ & $X=5$\\ \hline
				$F_{X|S}(\cdot|s_i,\rho)$ & 0.2  & 0.42 & 0.925 & 1 & 1 \\ \hline
				$F_{X|S}(\cdot|s_j,\rho)$ & 0  & 0.075 & 0.575 & 0.925 & 1 \\  \hline
			\end{tabular}
		\end{center}
	\end{table}
	We apply Lemma~\ref{lemma:ComputePiStar} using \autoref{tab:ex2F}. For $X = 1$, as $F_{X|S}(1|s_i,\rho) > F_{X|S}(1|s_j,\rho)$, we have
	$\triangle^*(1) = 	\min \Set{\triangle>0 \colon  F_{X|S}(1|s_i,\rho) \leq F_{X|S}(1+\triangle|s_j,\rho)} = 2$. In the same way, we can work out that $\triangle^*(2) = \triangle^*(3) = \triangle^*(4)$ and $\triangle^*(5) = 0$. So, $\sup_{x \in \X} \triangle^*(x) = 2$. We also compute the Kantorovich optimal plan $\pi^*$ in~\autoref{tab:ex2PiStar}. It can be seen that $\sup_{(x,x') \in \supp(\pi^*)}  |x-x'| = \sup_{x \in \X} \triangle^*(x) = 2$. 
	    \begin{table}[t]
		\caption{Kantorovich Optimal Transport Plan $\pi^*(x,x')$} \label{tab:ex2PiStar}
		\begin{center}
			\begin{tabular}{cccccc}
				\hline\hline
				& $X' = 1$ & $X'=2$ & $X'=3$ & $X'=4$ & $X'=5$\\ \hline
				$X= 1$ & 0  & 0.075 & 0.125 & 0 & 0\\ \hline
				$X= 2$ & 0  & 0 & 0.225 & 0 & 0 \\  \hline
				$X= 3$ & 0  & 0 & 0.15 & 0.225 & 0.125 \\  \hline
				$X= 4$ & 0  & 0 & 0 & 0 & 0.075 \\  \hline
				$X= 5$ & 0  & 0 & 0 & 0 & 0 \\  \hline
			\end{tabular}
		\end{center}
	\end{table}
\end{example}

\subsection{Interpretation of  \autoref{theo:SuffCondSab} and \autoref{theo:SuffCondSaperp}} \label{app:ex}
We prove \autoref{theo:SuffCondSab} and \autoref{theo:SuffCondSaperp} by Proposition~\ref{prop:ComputePiStar} based on the Kantorovich approach in Proposition~\ref{prop:Kantorovich}. 
For prior $P_{X|S}(\cdot | s_{a_i})$, the cumulative distribution function is 
\begin{align*}
	F_{X|S}(x|x_{a_i}) 
	&= \int_{-\infty}^{x} \Pr(f(D_{-i}) = m - a_i)  \prod_{j \in \N \colon j \neq i} \zeta_j  \dif m \\
	&= F (f(D_{-i}) = x - a_i) \prod_{j \in \N \colon j \neq i} \zeta_j, 
\end{align*}
where $F (f(D_{-i}) = x) = \int_{-\infty}^{x} \Pr(f(D_{-i}) = m) \dif m $. Then, 
\begin{align*}
	F_{X|S}&(x|x_{a_i})  - F_{X|S}(x - \triangle |x_{b_i}) \\
	& = \big( F (f(D_{-i}) = x - a_i)  - F (f(D_{-i}) = x - \triangle - b_i)     \big) \prod_{j \in \N \colon j \neq i} \zeta_j \\
	& > 0, \qquad \forall \triangle > a_i - b_i, \\
	F_{X|S}&(x|x_{a_i})  - F_{X|S}(x + \triangle |x_{b_i}) \\
	& = \big( F (f(D_{-i}) = x - a_i)  - F (f(D_{-i}) = x + \triangle - b_i)     \big) \prod_{j \in \N \colon j \neq i} \zeta_j \\
	& < 0, \qquad \forall \triangle > b_i - a_i, \\
\end{align*}
So, $\sup_{(x,x') \in \supp(\pi^*)}  |x-x'| = \sup_{x \in \X} \triangle^*(x)  = |a_i-b_i|$ for each pair of secrets $(s_{a_i}, s_{b_i})$ and therefore $\theta = \frac{\max_{i \in \N, \rho} | a_i - b_i |}{\epsilon}$ and \autoref{theo:SuffCondSab} proves. 
By Remark~\ref{rem:SabSaperp}, $\sup_{(x,x') \in \supp(\pi^*)}  |x-x'| = \sup_{x \in \X} \triangle^*(x)  = |a_i|$ for each pair of secrets $(s_{a_i}, s_{\perp_i})$ and therefore $\theta = \frac{\max_{i \in \N, \rho} | a_i|}{\epsilon}$ and \autoref{theo:SuffCondSaperp} proves.

\bibliographystyle{ACM-Reference-Format}
\bibliography{BIB}

\end{document}